\documentclass[12pt, letterpaper]{article}
\usepackage[font=small, margin=1.5cm]{caption}
\usepackage[top=1in, bottom=1.1in, left=0.7in, right=0.7in]{geometry}
 
\usepackage{amsfonts}
\usepackage{amsmath}
\usepackage{amssymb}
\usepackage{enumerate}
\usepackage{natbib}
\usepackage{color} 
\usepackage{array}
\usepackage{graphicx,epstopdf}
 \usepackage{float}
\usepackage{comment}
\usepackage{rotating}
\usepackage{url}
\usepackage{tikz}
 \usepackage{subfigure}
\usepackage{caption}
\usepackage{setspace}
\usepackage{hyperref}
\usepackage{mathtools}                      
\mathtoolsset{showonlyrefs=true}            
\newtheorem{theorem}{Theorem}

\newtheorem{example}[theorem]{Example}

\newtheorem{proposition}[theorem]{Proposition}
\newtheorem{remark}[theorem]{Remark}

\newenvironment{proof}[1][Proof]{\text{#1.} }{\ \rule{0.5em}{0.5em}}
\def\R{{\mathbb R}}        
\def\E{{\mathbb E}}        
\def\1{{\mathbf 1}}        

\def\L{{\mathcal L} \,}

\numberwithin{equation}{section}
\numberwithin{theorem}{section}

\title{Optimal Execution of Limit and  Market Orders\\ with Trade Director, Speed Limiter, and Fill Uncertainty  }
\author{Brian Bulthuis\thanks{KCG Holdings Inc., New York, NY 10006. Email: bbulthuis@kcg.com} \and Julio Concha\thanks{KCG Holdings Inc., New York, NY 10006. Email: jconcha@kcg.com} \and Tim Leung\thanks{Applied Mathematics Department, University of Washington, Seattle WA 98195. Email: {timleung@uw.edu}. \mbox{Corresponding author}. } \and Brian Ward\thanks{Industrial Engineering \& Operations Research (IEOR) Department, Columbia University, New York, NY 10027.Email: {bmw2150@columbia.edu}.} \ }

\date{\today}

\begin{document}
\maketitle

\abstract{We study the optimal execution of market and limit orders with permanent and temporary price impacts as well as  uncertainty in the filling of limit orders. Our continuous-time model incorporates a trade speed limiter and a trader director to provide better control on the trading rates. We formulate a stochastic control problem to determine the optimal dynamic strategy for trade execution, with  a quadratic terminal penalty to ensure complete liquidation. In addition, we identify conditions on the  model parameters  to ensure optimality of the controls and finiteness of the associated value functions. For comparison, we also solve the schedule-following  optimal execution problem that penalizes deviations from an order schedule. Numerical results are provided to illustrate the  optimal market and limit orders over time.  }

\newpage
\section{Introduction}
In the U.S. market, institutional investors own or manage a major share of the public equities.\footnote{In 2010, about 67\% of the U.S. equity market capitalization is owned/managed by institutional investors, as compared to 8\% in 1950, according to report by SEC Commissioner Luis A. Aguilar. Source: \url{https://www.sec.gov/News/Speech/Detail/Speech/1365171515808}.}  For many institutional investors, orders to buy or sell stocks can come in large sizes.  Such large orders are thought to be costly to implement as they create large, immediate demand (resp. supply) for a buy (resp. sell) order. A large buy (resp. sell) order may cause other traders to raise (resp. lower) their offered price because they perceive a change in value from the contra-orders of the institutional investor. This effect of additional implementation cost due to short-term liquidity demand is often called market impact.  Practitioners commonly view this to be a significant cost to track, avoid or minimize. Slow trading usually reduces  market impact by minimally affecting demand/supply over time. However, this strategy exposes the agency trader to the risk in the stock price movement and also the failure to complete buying/selling all desired units of stocks. Hence, it is the goal of many algorithmic traders to decide how to strategically trade off cost and risk for an institutional trade.

In this paper, we  analyze a continuous-time stochastic model for optimal execution using both  market and limit orders. Our paper extends the foundational market impact model of \cite{AC2000} by including limit orders with uncertain fill rates, and a speed limiter that penalizes overly large trading rates. Unlike \cite{AC2000}, our model has two order types, which can potentially have negative signs (i.e., submission of sell orders in a buy program or vice versa). Thus, we will need to consider how to keep trading signs to be positive. Additionally, we construct penalties that drive the optimal trading rates in the desired direction of complete liquidation. Specifically, we include a non-liquidation terminal penalty as well as a penalty, called the trade director, to push trading rates for market and limit orders to be in the same direction. In combination, these three penalties (non-liquidation, trade director and speed limiter) force the algorithm to trade to full liquidation, while simultaneously tailoring the sign and magnitude of the trading rates.  

Our main objective is to investigate how to optimally allocate aggressive (market) and passive (limit) orders over time. The aggressive nature of market orders means they tend to have higher market impact. On the other hand, limit orders are lower cost in both price and market impact, but they are less likely to fill.  The  fill uncertainty is modeled  as an affine function of the trading rate in limit orders, and is reflected in an additional diffusion term that is correlated with the stock price process.

Under our framework, the optimal liquidation problem is cast as a stochastic optimal control problem. The associated nonlinear HJB equation can be simplified to a system of linear ODEs. We then explore the special cases of constant and linear uncertainty to examine the properties of the solutions and analyze the corresponding explicit optimal strategies. Among our results, we characterize the trader's \textit{buy-sell boundary}, which is a time deterministic function governing when order types are non-negative. We also consider an alternative model to incorporate a  benchmark trading schedule, and  show that our model is capable of generating optimal strategies that very closely follow any given schedule while minimizing trading cost. This allows us to understand a tradeoff between following schedule and seeking profits, and thus, evaluate how profitable it may be to deviate from schedule. Throughout the paper, we provide numerical results to illustrate the optimal liquidation strategies in various settings. 

Among our findings, we derive the trading rates explicitly and further examine the conditions under which they are non-negative over the trading horizon. In addition, we also introduce the critical time span for an order placement program. If the trading horizon exceeds this critical time, the trader has too much time to explore profitable opportunities in trading the stock. This leads to an interesting and intuitive trade-off: the non-liquidation penalty must be increased in order to prevent the trader from using time to explore unbounded profits, resulting in over trading. 

There are a number of related studies on  optimal liquidation with similar basic settings as in \cite{AC2000}, though liquidation with both market and limit orders has only come to the forefront of the algorithmic trading literature in recent years. A recent paper by \cite{TaiHoACF} extends the Almgren-Chriss framework to include uncertain order fills of a single order type. Our paper extends their model to include both market and limit orders, along with additional constraints and penalties to guide the trade direction and limit order size. In the case of infinite uncertainty, our framework captures their model as a special case, with our optimal market order rate coinciding with theirs. \cite{SeblimitmarketQF2015} proposed a model for optimal execution with market and limit orders that uses jump processes and optimal multiple stopping to determine the optimal market order placement time. They too penalize deviations from a schedule. For a specific example of schedule following, we refer to a recent study by \cite{sebVWAP},  which derives  a  closed-form   optimal strategy that  follows a volume weighted average price (VWAP) schedule.



 Our paper is also related to the literature on the optimal market making problem. Market making involves simultaneously determining the prices and quantities to buy and sell a stock. The market maker receives the spread in exchange for the risk of holding a position. \cite{AvellandaStoikov2008}   apply indifference pricing techniques to find the optimal quotes for a risk-averse investor trading over a finite period. \cite{GuilbaudPham2013}   study a market making problem via the optimal placement of limit orders as well as using market orders to balance inventory risk. In these studies, there are no constraints on the signs of orders because a market maker typically  places buy and sell orders simultaneously. For more related studies on algorithmic trading and market microstructure, we refer to the books by \cite{lehalleBook} and by \cite{sebBook}.

The rest of the paper is organized as follows. In Section \ref{sec:model_order_fills}, we formulate a stochastic control problem for     optimal order execution under a sell program. In Section \ref{sec:affine_unc}, we solve the problem under the assumption of affine order fill uncertainty. Then, we discuss a number of important properties of the solutions and trading scenarios under constant uncertainty in Section \ref{sec:constant_unc} and under linear uncertainty in Section  \ref{sec:linear_unc}. In Section \ref{sec:schedule_following}, we adapt our model to penalize strategies that deviate from a pre-specified schedule for share holdings. Section \ref{sec:conclude} concludes the paper. 

\section{Optimal Order Type Selection}\label{sec:model_order_fills} 
Throughout this paper, we take the perspective of a sell program. The mathematics for a buy program is completely analogous. We first present the formulation of our optimal execution model, then derive the optimal strategies. We will further discuss their properties in Sections \ref{sec:constant_unc} and \ref{sec:linear_unc} under constant and linear limit order uncertainty, respectively.  

 In the background, we fix a probability space $(\Omega, \mathcal{F},  \mathbb{P})$, and a finite trading horizon $[0,T]$. Our model involves two stochastic  controls: (i) the trading rate of market orders, $v_t$ and (ii)  the trading rate of limit orders, $L_t$, over time $t\in [0,T]$. The trader's stock holdings, denoted by $x_t$ at time $t$, is depleted by the trading rates, $v_t$ and $L_t$. To capture the uncertainty of limit order fills, there is an additional diffusion term $m(L_t)dZ_t$, where $m$ is a deterministic function of the limit order trading rate and $Z$ is a standard Brownian motion.   Precisely, the trader's position satisfies the SDE:
\begin{equation}\label{eq:SDEx}
dx_t=-v_t dt+\left(-L_t dt+m(L_t)dZ_t\right).
\end{equation}
The stock price $S$ and transacted price $\widetilde{S}$ follow a generalized version of \cite{AC2000} dynamics: 
\begin{equation}\label{eq:S_price}
\begin{aligned}
dS_t&=\gamma dx_t +\mu dt +\sigma dW_t,\\
\widetilde{S}_t&=S_t+h(v_t,L_t).
\end{aligned}
\end{equation}
In other words, the price $S$ follows an arithmetic Brownian motion with drift $\mu\in \R$ and volatility $\sigma>0$, along with a linear permanent impact with coefficient $\gamma>0$. The transacted price reflects a temporary impact, $h(v_t,L_t)$, which is a function of the current trading rates of market and limit orders. As we seek closed-form solutions, we assume that the temporary impact is affine, i.e. $h(v_t,L_t)=-\eta_0-\eta_1 v_t-\eta_2 L_t$, with constants $\eta_0,\eta_1,\eta_2>0$. The two standard Brownian motions, $Z$ and $W$, are correlated with an instantaneous correlation parameter $\rho\in(-1,1)$. The investor's information flow is modeled by the filtration $\mathbb{F}$ generated by $(Z,W)$, and all admissible strategies must be  $\mathbb{F}$-adapted.

We typically expect $\rho m(L)<0$ so long as $L\ge0$. This parameter sign choice will give rise to an \emph{adverse selection} effect that is often considered in algorithmic trading models. Adverse selection is an implicit cost incurred when, for example, the stock price rises just after making a sale. It would have been better for the trader to wait for the price rise before selling the stock because then she would have realized an extra profit. To see this, first consider the case with $\rho>0$. Then, if $\Delta Z_t>0$ is observed over some small time period $\Delta t$, we typically see $\Delta W_t>0$, and thus $\Delta S_t>0$, an increase in the stock price. For the model to incorporate an adverse selection effect, we want the trader to make an excess sale in limit orders. This will occur if $m(L_t)<0$, for then $m(L_t)\Delta Z_t<0$ and $\Delta x_t<0$. Reversing the signs and starting with $\rho<0$, we find that we need $m(L_t)>0$ for $\Delta x_t<0$. The two cases can be summarized by the restriction $\rho m(L)<0$. In our implementation, we will choose parameters that satisfy this adverse selection criterion, though our model also works when it does not hold. 

As a standard performance metric, the  profit-and-loss (PNL) of any trading strategy is defined as
\begin{align}
\Pi_T&:=x_T\left(S_T-S_0\right)+\int_0^T\left(S_0-\widetilde{S}_u\right)dx_u\\
&=x_T S_T-x_0 S_0-\int_0^T \widetilde{S}_u dx_u.\label{PNL}
\end{align}
The first two terms in  \eqref{PNL}   measure the change in fair value of the portfolio as measured by the stock  price. The third term is the revenue/cost of trading: it measures at each point in time the change in position multiplied by the transacted price $\widetilde{S}_t$ at that time. In a sell program, the infinitesimal amount of shares sold (represented by $-dx_t$) tends to be positive. Thus, the integral $-\int_0^T \widetilde{S}_u dx_u$ can be interpreted as revenue of sale. 

When we plug in the dynamics given by equations \eqref{eq:SDEx} and \eqref{eq:S_price}, the PNL becomes
\begin{equation}\label{eq:PNL_calculated}
\begin{aligned}
\Pi_T&=\frac{\gamma}{2}\left(x_T^2-x_0^2\right)+\int_0^T\left[\mu x_u+\rho\sigma m(L_u) +\frac{\gamma}{2}m^2(L_u)+h(v_u,L_u)(v_u+L_u)\right]du\\
&+\int_0^T \sigma x_u dW_u -\int_0^T h(v_u,L_u) m(L_u) dZ_u.
\end{aligned}
\end{equation}


We also incorporate three model features: (i) quadratic terminal penalty, (ii) trade director, and (iii) trade speed limiter. First, we must ensure that the position is actually liquidated. To that end, we add a quadratic penalty term $\bar{f}(x)=-\beta x^2$, so that anything but complete liquidation is undesirable. This is the same penalty as was studied in \cite{TaiHoACF}. The second penalty is introduced to penalize placing buy side market orders and sell side limit orders (or vice versa) simultaneously. Such order placement would be as if the trader were buying/selling the stock to herself. This penalty is called the trade director. The goal is to encourage orders to have the same sign or equivalently, $v_tL_t\ge0$. In combination with the non-liquidation penalty, we expect that only the choice $v_t\ge0$, $L_t\ge0$. We introduce   the penalty integral $\int_0^T \alpha v_u L_u du$, with the Lagrange multiplier $\alpha\ge0$. 

We also add a speed limiter to prevent  the trader from trading too quickly with either order type. Suppose that management has set two  trading speed caps $r_{1}, r_{2}>0$, for market and limit orders, respectively. To encourage the trader to satisfy the constraints: $-r_{1}\le v_t \le r_{1}$ and $-r_{2}\le L_t \le r_{2}$, for all $t$, or equivalently, $v_t^2 \le r_{1}^2=:R_1$ and  $L_t^2 \le r_{2}^2=:R_2$,   we introduce the penalty term  $\int_0^T \beta_1(R_1-v_u^2)+\beta_2(R_2-L_u^2) du$, with the Lagrange multipliers $\beta_1, \beta_2\ge0$. The use of quadratic penalties such as the ones described here is also found in other classes of trading problems, such as hedging  via risk minimization with constraints (see, for example, \cite{lee2008}). 

Incorporating these three features, and utilizing equation \eqref{eq:PNL_calculated}, we now write  the compensated PNL as
\begin{align}
\widehat{\Pi}_T &=\Pi_T+\int_0^T\!\big[ \alpha v_uL_u +  \beta_1(R_1-v_u^2)+\beta_2(R_2-L_u^2) \,\big]\,du+\bar{f}(x_T)\label{PNL_comp}\\
 &=\frac{\gamma}{2}\left(x_T^2-x_0^2\right)+\bar{f}(x_T)+\int_0^T \sigma x_u dW_u\\
 &\quad  - \int_0^T h(v_u,L_u)m(L_u)dZ_u+\int_0^Tg(x_u,v_u,L_u)du,\label{Pi22}
\end{align}
where we have defined 
\begin{equation}
\begin{aligned}
\bar{f} (x) &=\, -{\beta} x^2,\\
g(x,v,L)&=\, \mu x +\rho \sigma m(L)+\frac{\gamma}{2}m^2(L)+h(v,L)(v+L)\\&+\alpha vL+\beta_1(R_1-v^2)+\beta_2(R_2-L^2).\label{fun_g}
\end{aligned}
\end{equation}
  The trader's objective is to maximize the expectation of the compensated PNL $\widehat{\Pi}_T$ in  \eqref{Pi22} by choosing trading rates for market and limit orders.  This   leads to  the value function
\begin{equation}\label{eq:opt_control_prob}
\begin{aligned}
&V(t,x):=\underset{(v_u,L_u)_{t\le u\le T}}{\sup}\,\E\left[\,\bar{f}(x_T)+\frac{\gamma}{2}x_T^2+\int_t^T g(x_u,L_u,v_u) du\,\Big|\,x_t=x\,\right]-\frac{\gamma}{2}x^2.
\end{aligned}
\end{equation}
We will study the nonlinear HJB PDE problem associated with the stochastic control problem  \eqref{eq:opt_control_prob}:
\begin{equation}\label{HJBeq_TaiHo}
	\begin{aligned}
 V_t+\underset{v,L}{\sup}\left[-(v+L)V_x+\frac{1}{2}m^2(L)V_{xx}+g(x,v,L)\right] &= 0, \quad (t,x) \in [0,T) \times \R,\\
 V(T,x)&=\bar{f}(x)+\frac{\gamma}{2}x^2, \quad x \in \R,
\end{aligned}
\end{equation} and solve for the value function $V$ and the corresponding optimal trading strategies.


\section{Affine Uncertainty of Limit Orders}\label{sec:affine_unc}
We present here an analytic  solution to the stochastic  control problem \eqref{eq:opt_control_prob} when the fill uncertainty function is affine in the limit order trading rate.  In subsequent sections, we will analyze the  special cases with constant uncertainty and linear uncertainty, and highlight their distinct features.

Following our model formulation, we now  set the fill uncertainty and temporary price impact to be the affine functions:
\begin{align}\label{AFFINE}
m(L)=m_0+m_1L, \quad \text{ and } \quad  h(v,L)=-\eta_0-\eta_1 v -\eta_2 L. 
\end{align} 
In this case, the HJB PDE problem in \eqref{HJBeq_TaiHo}  becomes 
\begin{equation}\label{eq:HJBeq}
\begin{aligned}
V_t+\mu x+\frac{m_0^2}{2}(V_{xx}+\gamma)+\rho\sigma m_0+\beta_1R_1+\beta_2R_2+\underset{v,L}{\sup}\,\big\{J(t,x,v,L)\big\}=0,
\end{aligned}
\end{equation}
for all $(t,x)\in\left[0,T\right)\times\R$, with  the terminal condition $V(T,x)=\left(\frac{\gamma}{2}-\beta\right)x^2$ for all $x\in\R$. Note that $J$ in \eqref{eq:HJBeq}
 is defined as
\begin{equation}\label{eq:Jdefn}
\begin{aligned}
J(t,x,v,L):&=-(v+L)V_x+\frac{1}{2}(m_0+m_1L)^2V_{xx}+g(x,v,L)\\
&-\mu x-\frac{m_0^2}{2}(V_{xx}+\gamma)-\rho\sigma m_0-\beta_1R_1-\beta_2R_2\\
&=\left(-\eta_0-V_x +m_0m_1(V_{xx}+\gamma)+\rho \sigma m_1\right)L+\left(-\eta_2-\beta_2+\frac{m_1^2}{2}(V_{xx}+\gamma)\right)L^2 \\
&\left(-\eta_0-V_x\right) v+\left(-\eta_1-\beta_1 \right)v^2+\left(\alpha-\eta_1-\eta_2\right)vL
\end{aligned}
\end{equation}
where  $g(x,v,L)$ is defined in \eqref{fun_g}.

\begin{proposition}\label{prop:optcond} Under the affine uncertainty model, if the value function $V(t,x)$ satisfies the  \textit{second-order condition}:
\begin{align}
	m_1^2 \,V_{xx}(t,x)< {{C}-\gamma m_1^2}, \qquad \forall (t,x),\label{optimalcond}
\end{align} where \begin{align}\label{CC}C:= \frac{4(\eta_1\eta_2+\beta_1\beta_2+\eta_1\beta_2+\eta_2\beta_1)-\left(\eta_1+\eta_2-\alpha\right)^2}{2(\eta_1+\beta_1)},\end{align}
  then the optimal liquidation problem has finite optimum and there exist unique globally optimal controls $(v^*, L^*)$.
\end{proposition}
\begin{proof}
 The function $J$  in \eqref{eq:HJBeq} is a bivariate quadratic function in $v$ and $L$, so the first-order conditions for the supremum in  \eqref{eq:HJBeq} are  a pair of linear equations.   To facilitate the presentation, we define   
\begin{align}
\psi(t,x)&:=m_1^2(V_{xx}(t,x)+\gamma),\label{psidef}\\
\textbf{A}&:= \left( \begin{array}{cc}
-2\eta_1-2\beta_1 & \alpha-\left(\eta_1+\eta_2\right) \\
\alpha-\left(\eta_1+\eta_2\right) & \psi(t,x)-2\eta_2-2\beta_2  \end{array} \right). \label{matrixA}
\end{align} 
where $\textbf{A}$ is the Hessian matrix  of $J$ w.r.t. $v$ and $L$ and its dependence on $t$ and $x$ has been suppressed. Then the first-order conditions    are
\[ \textbf{A}\left(\begin{array}{c} v\\L \end{array}\right)=\left(\begin{array}{c} V_x+\eta_0\\ V_x+\eta_0-m_0m_1\left(V_{xx}(t,x)+\gamma\right)-\rho\sigma m_1 \end{array}\right).\]
Furthermore, if $\textbf{A}$ is \emph{negative definite} for all $(t,x)$, then the optimal execution problem has a finite solution uniquely determined by these first-order conditions.

To check the negative definiteness of the Hessian, we calculate its eigenvalues. The characteristic equation is
\begin{align}
\det(\textbf{A}-u\textbf{I})=(2\eta_1+2\beta_1+u)(2\eta_2+2\beta_2-\psi(t,x)+u)-\left(\eta_1+\eta_2-\alpha\right)^2=0.
\end{align}
The solutions correspond to the two  eigenvalues, $u_{\pm}$, given by \begin{align}
u_{\pm}=-\eta_1-\eta_2-\beta_1-\beta_2+\frac{\psi(t,x)}{2}\pm\sqrt{\left(\frac{\psi(t,x)}{2}+\eta_1-\eta_2+\beta_1-\beta_2\right)^2+\left(\eta_1+\eta_2-\alpha\right)^2}.
\end{align} 
For $u_+<0$, we must have
\begin{align}
\psi(t,x)<\frac{4(\eta_1\eta_2+\beta_1\beta_2+\eta_1\beta_2+\eta_2\beta_1)-\left(\eta_1+\eta_2-\alpha\right)^2}{2(\eta_1+\beta_1)}=:C\label{condineq}
\end{align} 
for all $(t,x)$, which is equivalent to condition \eqref{optimalcond}. Next, since $u_-\le u_+$, $u_+<0$ implies $u_-<0$. As a result, if the controls $v$ and $L$ that solve the first-order conditions yield a solution $V$ to the HJB equation such that \eqref{condineq} holds (or equivalently \eqref{optimalcond} after rearrangement), rearrangement, then the optimal liquidation problem has a finite optimum. Moreover, the first-order conditions have a unique solution. Therefore, there exist unique globally optimal controls.\end{proof}


%
\begin{remark}
We will consider in Section \ref{sec:constant_unc} the special  case with  constant uncertainty ($m_1=0$). Then, condition \eqref{optimalcond}   in Proposition \ref{prop:optcond} no longer depends on $(t,x)$ and can be simplified as $C>0$. In Section \ref{sec:lagrange} we discuss the practical consequences  of this condition. As we will derive the value function $V$ explicitly, we can directly verify condition \eqref{optimalcond}.
\end{remark}

Given the above proposition, we know there exist unique globally optimal controls under the condition \eqref{optimalcond}. We proceed now to find those controls. For simplicity, we denote the determinant of A by $\Delta(t,x)$. This is given by: 
\begin{align}
\Delta(t,x)=(2\eta_1+2\beta_1)(2\eta_2+2\beta_2-\psi(t,x))-\left(\eta_1+\eta_2-\alpha\right)^2.
\end{align}
With this notation, the solutions to the first-order conditions are
\begin{equation}\label{eq:optimzers}
\begin{aligned}
v_t^*&=\frac{(m_1^2(V_{xx}+\gamma)+\eta_1-\eta_2-2\beta_2-\alpha)(V_x+\eta_0)-(\eta_1+\eta_2-\alpha)(m_0m_1(V_{xx}+\gamma)+\rho\sigma m_1)}{\Delta},\\
L_t^*&=\frac{(\eta_2-\eta_1-\alpha-2\beta_1)(V_x+\eta_0)+2(\eta_1+\beta_1)(m_0m_1(V_{xx}+\gamma)+\rho\sigma m_1)}{\Delta}. 
\end{aligned}
\end{equation}
Substituting $(v^*, L^*)$    into the HJB equation, we arrive at  the following nonlinear PDE:
\begin{equation}\label{eq:Affine_PDE}
\begin{aligned}
0&=V_t+\mu x+\frac{m_0^2}{2}(V_{xx}+\gamma)+\rho\sigma m_0+\beta_1R_1+\beta_2R_2+\frac{(V_x+\eta_0)^2}{4(\eta_1+\beta_1)}\\
&+\frac{2(\alpha+\beta_1+\beta_2)-C}{2\Delta}\left[V_x+\eta_0+\frac{(\eta_2-\eta_1-\alpha-2\beta_1)(m_0m_1(V_{xx}+\gamma)+\rho\sigma m_1)}{2(\alpha+\beta_1+\beta_2)-C}\right]^2,
\end{aligned}
\end{equation}
subject to the terminal condition $V(t,x)=\left(\frac{\gamma}{2}-\beta\right)x^2$. 

The quadratic terminal condition suggests the ansatz $V(t,x)=a(t)x^2+b(t)x+c(t)$. This ansatz will solve equation \eqref{eq:Affine_PDE} as long as the coefficient functions solve the following system of first-order ODEs:
\begin{equation}\label{eq:ODE_system}
\begin{aligned}
0&=a^\prime(t)+\frac{2(\alpha+\beta_1+\beta_2)-m_1^2(2a(t)+\gamma)}{(\eta_1+\beta_1)(C-m_1^2(2a(t)+\gamma))}\,a^2(t), \qquad\qquad\qquad a(T)=\frac{\gamma}{2}-\beta, \\
0&=b^\prime(t)+\mu+\frac{2(\alpha+\beta_1+\beta_2)-m_1^2(2a(t)+\gamma)}{(\eta_1+\beta_1)(C-m_1^2(2a(t)+\gamma))}a(t)(b(t)+\eta_0), \\
&+\frac{(\eta_2-\eta_1-\alpha-2\beta_1)(m_0m_1(2a(t)+\gamma)+\rho\sigma m_0)}{(\eta_1+\beta_1)(C-m_1^2(2a(t)+\gamma))}a(t),\qquad\qquad b(T)=0,\\
0&=c^\prime(t)+\frac{m_0^2}{2}(2a(t)+\gamma)+\rho\sigma m_0+\beta_1R_1+\beta_2R_2\\
&+\frac{2(\alpha+\beta_1+\beta_2)-m_1^2(2a(t)+\gamma)}{4(\eta_1+\beta_1)(C-m_1^2(2a(t)+\gamma))}(b(t)+\eta_0)^2\\
&+\frac{(\eta_2-\eta_1-\alpha-2\beta_1)(m_0m_1(2a(t)+\gamma)+\rho\sigma m_0)}{2(\eta_1+\beta_1)(C-m_1^2(2a(t)+\gamma))}(b(t)+\eta_0)\\
&+\frac{(\eta_2-\eta_1-\alpha-2\beta_1)^2(m_0m_1(2a(t)+\gamma)+\rho\sigma m_0)^2}{8(\eta_1+\beta_1)(C-m_1^2(2a(t)+\gamma))(2(\alpha+\beta_1+\beta_2)-C)},\qquad\qquad c(T)=0.\\
\end{aligned}
\end{equation} These ODEs can be solved numerically successively. An analytic strategy would be to first solve for $a(t)$ by separation of variables. Then we can plug $a(t)$ into the ODE for $b(t)$. The resulting  first-order linear inhomogeneous ODE is readily solved. In turn, given $a(t)$ and $b(t)$, the third ODE is separable and solved directly by integration.   In the next sections, we will solve these equations analytically or numerically to present  the solutions and associated trading strategies.

Given that $V(t,x)=a(t)x^2+b(t)x+c(t)$, we can now simplify the second-order condition \eqref{optimalcond} and express it in terms of $a(t)$ and  model parameters:
\begin{equation}
	2m_1^2 a(t)< {{C}-\gamma m_1^2}, \qquad \forall (t,x).
\end{equation}
At    this point, $a$ satisfies an implicit equation, so  it is not possible to simplify the condition further. Nevertheless, under constant uncertainty, it will simplify significantly  as $m_1=0$. The linear uncertainty model will also allow for an analytical solution for $a(t)$.

\subsection{Numerical Illustration}\label{sec:numerics_first}
Before displaying the numerical implementation of our model, we discuss parameter choices. For this paper, we are thinking about trades that will be executed within a few hours at very high speeds. Therefore, it makes sense  to consider \emph{seconds} as the time interval and we use $T=3,600$ seconds in all simulations. For other parameters, we assume 6.5 trading hours per day and 252 trading days per year and select parameter values of the same magnitude as those found in \cite{AC2000}:
\begin{alignat}{2}
 S_0&= 40\,\$/\text{share},&\quad x_0 &= 10,000\,\text{share},\\
\sigma&= 0.005\,(\$/\text{share})/\text{sec.}^{0.5}, &\quad \mu &= 10^{-6}\,(\$/\text{share})/\text{sec.},\\
\gamma&= 2.5\cdot10^{-7}\,\$/\text{share}^2,&\quad \eta_0 &= 0.05\,\$/\text{share},\\
\eta_1 &= 0.1\,(\$/\text{share})/(\text{share}/\text{sec.}),&\quad
\eta_2 &= 0.08\,(\$/\text{share})/(\text{share}/\text{sec.}).
\end{alignat}
Note that converting with 6.5 trading hours per day and 252 trading days per year, we find that this value of $\sigma$ and initial stock price correspond to approximately $30\%$ annual volatility. The value of $\mu$ corresponds to approximately $15\%$ annual growth.

A small correlation value, $\rho=-0.2$, along with positive values of $m_0$ and $m_1$, are selected to include the  adverse selection effect. The uncertainty parameters are chosen as $m_0=p_0 {x_0}{{T}^{-0.5}}$, and $m_1=p_1\sqrt{T}$ for some (unitless) constants $p_0$ and $p_1$. If we want to compare the constant and linear uncertainty models then $p_0=p_1$ correspond to similar coefficients on the Brownian motion driving $x_t$. This is because $m_1$ will multiply $L_t$, which will roughly be of order $ {x_0}/{T}$.  Under  constant or linear uncertainty we set $p_0=p_1=0.1$,  whereas under affine uncertainty  we choose $p_0=p_1=0.05$ to  split the uncertainty between the two components. 

The terminal liquidation penalty will be $\beta=10^{-3}\,\$/\text{share}^2$. For the speed limiter we choose $\beta_1=5\cdot10^{-4}\,\$/(\text{share}/\text{sec.})^2$, and $\beta_2=10^{-4}\,\$/(\text{share}/\text{sec.})^2$. These are small penalties so they will not over-reduce the trade speed and they are of the same order as the terminal liquidation penalty. Finally, the trade director penalty is chosen as $\alpha=0.15$. In a subsequent section, we will find an allowable interval for $\alpha$. For the above parameter values, 0.15 is in the middle of the interval, just to the left of the center so it does not over-encourage high trade speeds. Throughout our numerical simulations, these will be the baseline parameters. We highlight the effects of each parameter by changing them from these baseline values in our numerical examples.

In  Figure \ref{fig:lin_and_const}, we plot  the trading rates for models with constant uncertainty, linear uncertainty, and zero uncertainty in the limit order fills. To ensure   a fair comparison, the Brownian motions generating the asset dynamics and the uncertainty in position are the same across all simulations. Only the functional forms of the trading rates change. Note that   both rates are positive for all three models. Furthermore, the strategies are primarily dominated by limit orders. This is preferred in practice because limit orders tend to be lower impact and lower cost. Indeed, that is the selected parameter choice. For the most part, over the life of the 1 hour sell program, the trading rates  are almost constant when there is no uncertainty, and remain quite  stable in the uncertainty cases. As we approach the execution horizon the rates become more unstable: the algorithm reacts to small moves to capitalize on low cost opportunities to liquidate the asset.

Nevertheless, over the life of a trade, the linear uncertainty trading rates are typically more stable and lower than the constant uncertainty trading rates.  To see this, notice that   the linear uncertainty can be avoided by choosing $L_t=0$, but this is not possible in the constant uncertainty case.  In this sense, there is less uncertainty for a linear uncertainty model versus a constant uncertainty model. This helps explain   the ordering of the trading rates. As   the uncertainty diminishes to no uncertainty, the future is more predictable, so the trader does not need to resort to overtrading. Finally, note that the  trading rates under linear and constant uncertainties tend to spike up significantly towards the end. The reason is that the non-liquidation penalty dominates all other costs, and thus, trading is sped up to achieve full liquidation. Again, even in this highly volatile period, linear uncertainty results in  lower trading rates than constant uncertainty. Interestingly, it is observed that the market order rate exceeds the limit order rate for the first time near the end of the trading horizon under linear uncertainty. The same does not occur under  constant uncertainty.  This is because under linear uncertainty limit order fill uncertainty can be eliminated by avoiding placing such orders, and the trader who wants to achieve full liquidation turns to market orders, which do not have fill uncertainty. 

Our method encourages non-negativity of trading rates, and  the trader tends to maintain $v_t>0$ and $L_t>0$ for the majority of the sell-program. However, this need not be the case in general and trading rates have the potential to be negative. Later, in section \ref{sec:cons_strats} we will show that non-negative trading rates can be guaranteed in special cases.


\begin{figure}[H] \centering
	\subfigure[]{\includegraphics[trim=10   0  32  8,clip,width=3in]{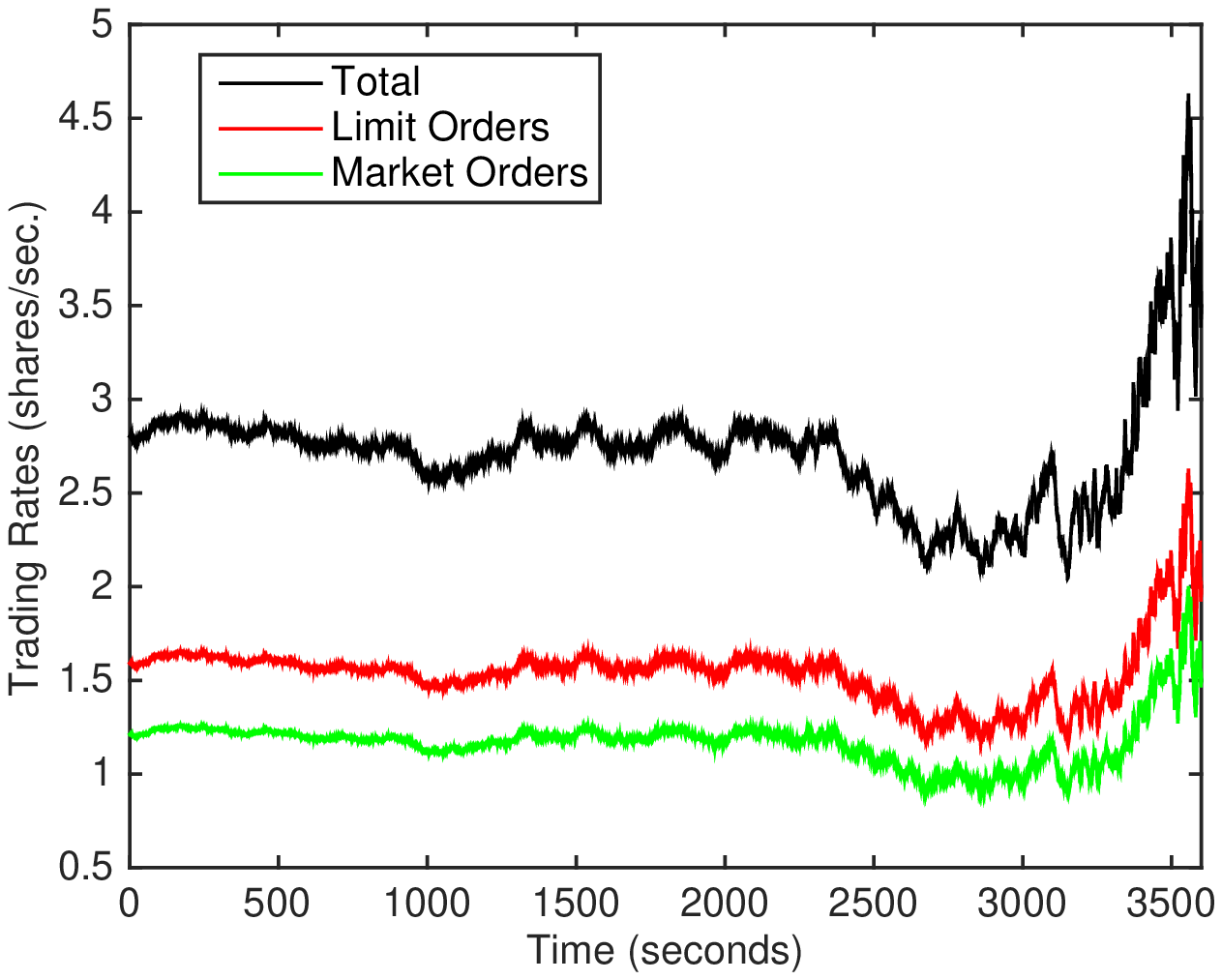}}
	\subfigure[]{\includegraphics[trim=10   0  32  8,clip,width=3in]{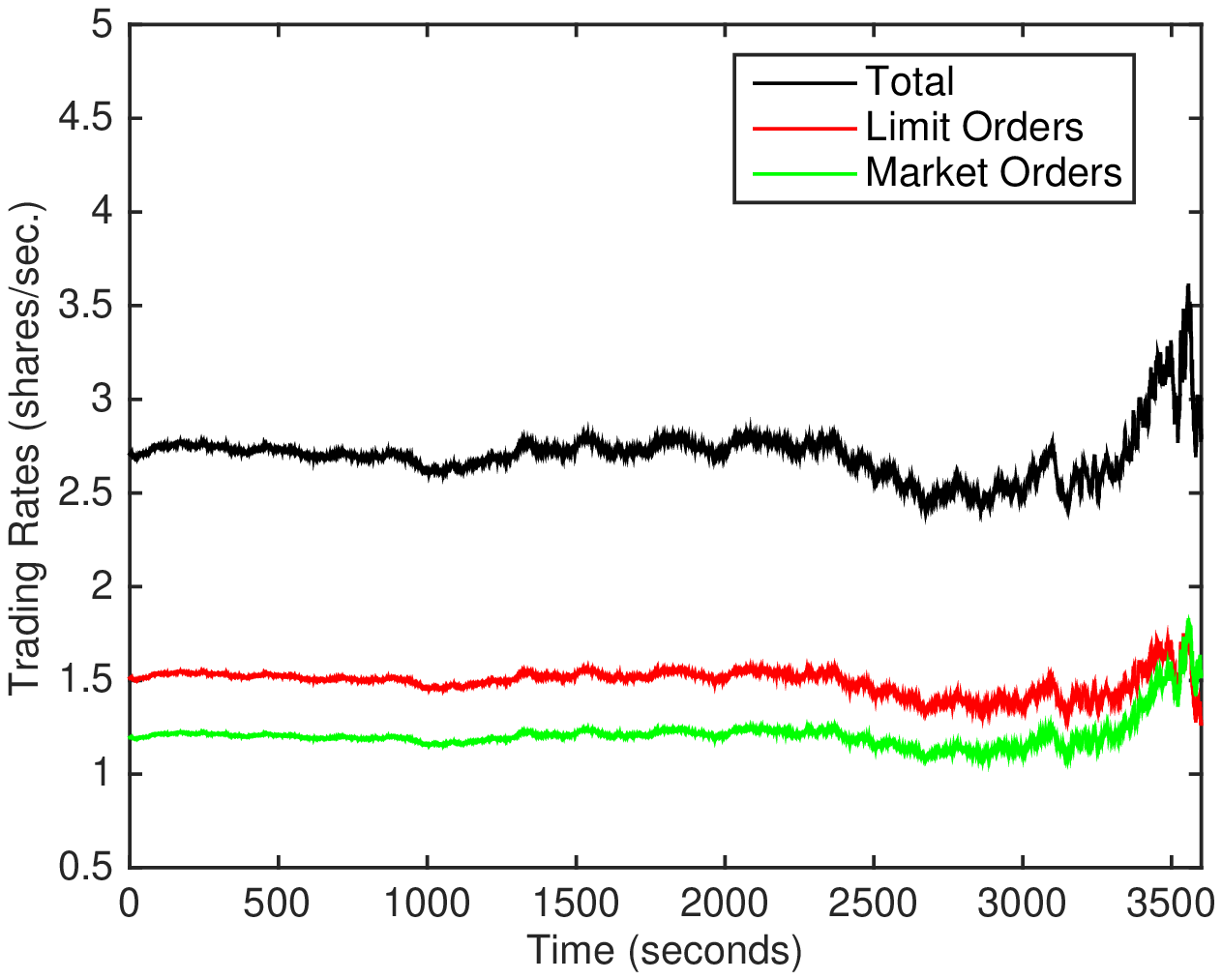}}
	\subfigure[]{\includegraphics[trim=10   0  32  8,clip,width=3in]{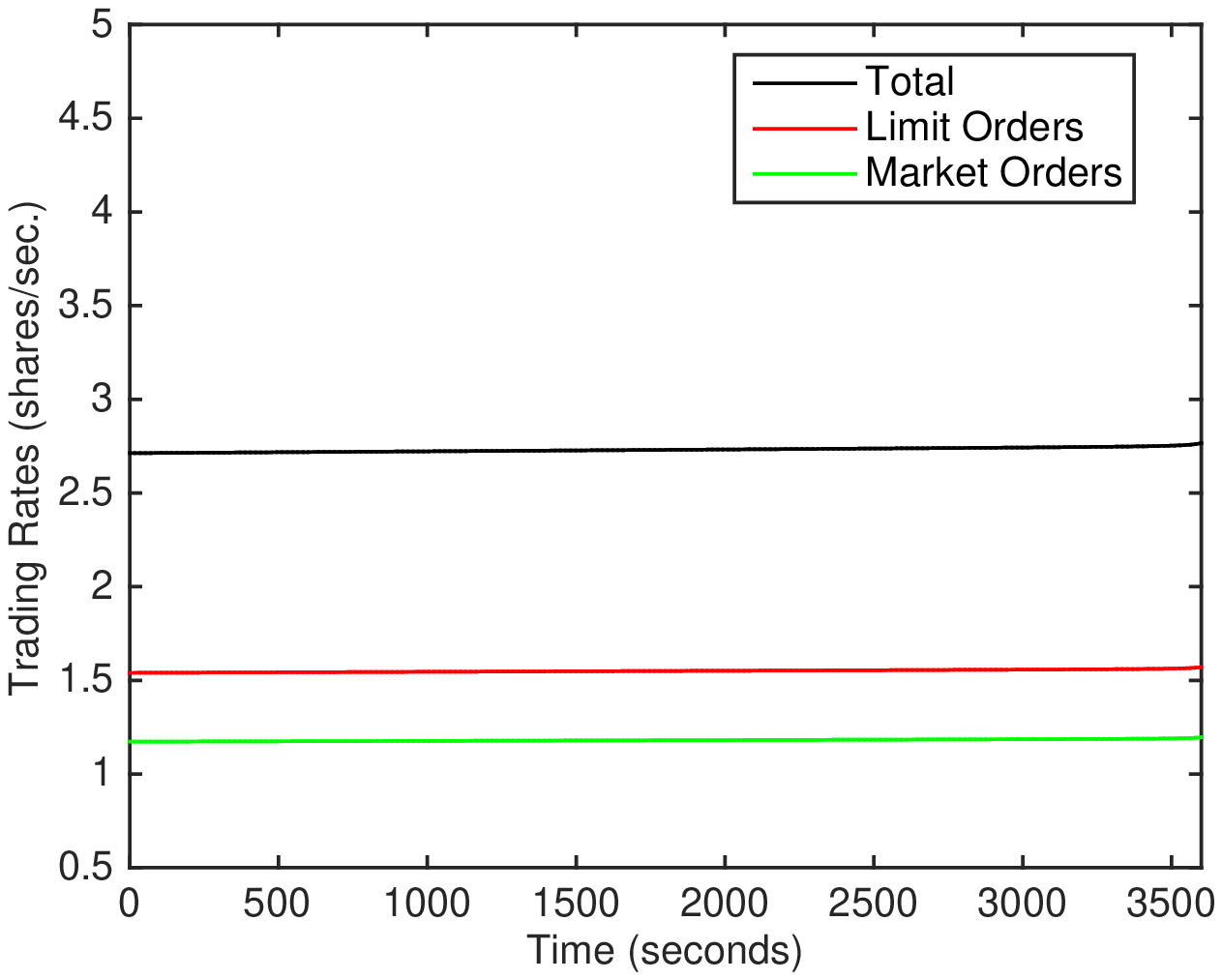}}
	\caption{\small{Sample paths of the  optimal trading rates under (a) constant uncertainty ($m_0=16.\bar{6}$), (b) linear uncertainty ($m_1=6$), and (c) no uncertainty. Final positions are 217.25, 203.39, and 168.75 shares, respectively. The parameters are $x_0=10,000$, $T=3,600$, $\beta=10^{-3}$, $\eta_0=0.05$, $\gamma=2.5\cdot10^{-7}$, $\alpha=0.15$, $\beta_1=5\cdot10^{-4}$, $\beta_2=10^{-4}$, $\eta_1=0.1$, $\eta_2=0.08$, $\rho=-0.2$, $\mu=10^{-6}$, and $\sigma=0.005$.}}\label{fig:lin_and_const}
\end{figure}

\begin{figure}[H] \centering 
	\subfigure{\includegraphics[trim=10   0  32  8,clip,width=3in]{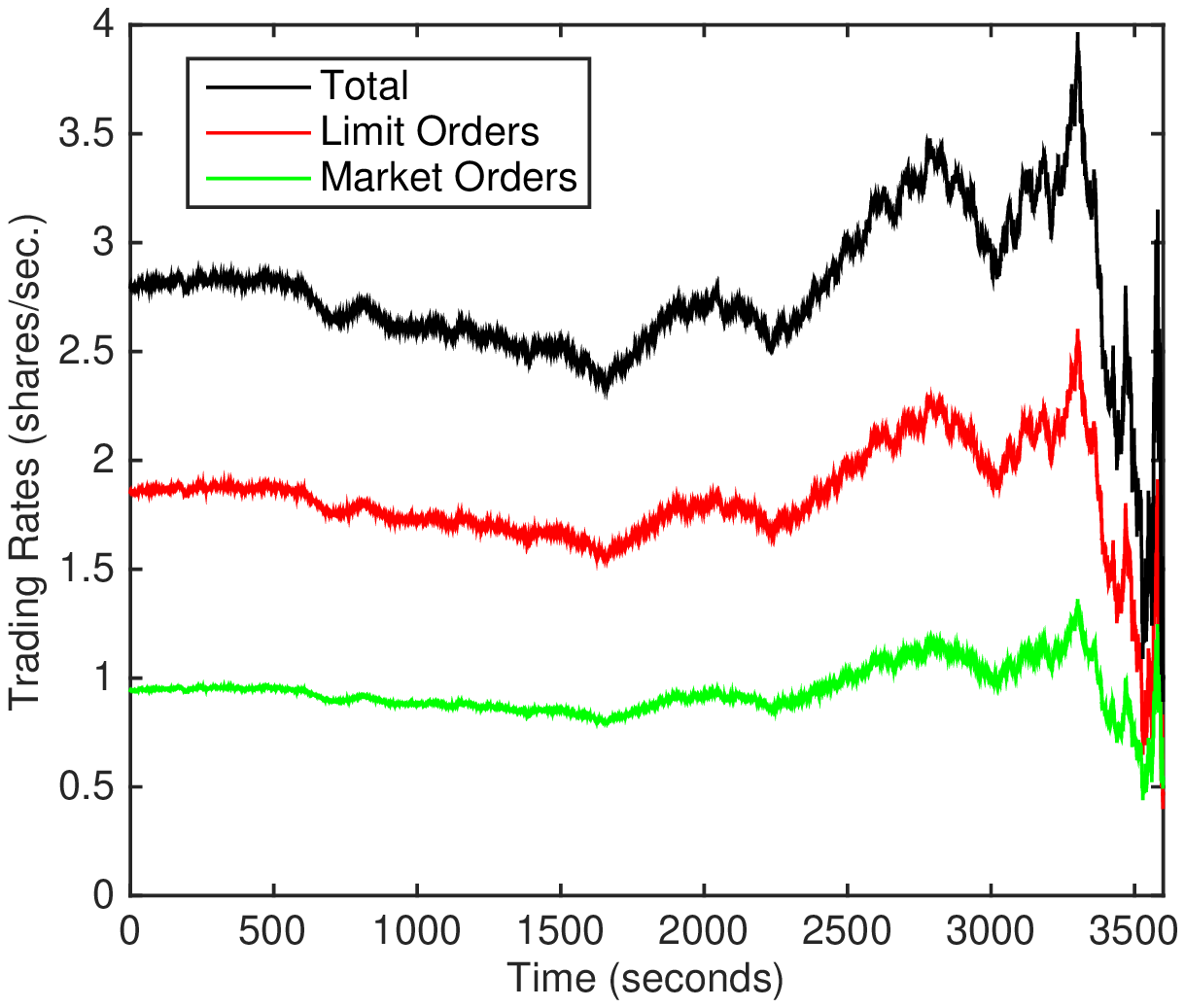}}
	\subfigure{\includegraphics[trim=10   0  32  8,clip,width=3in]{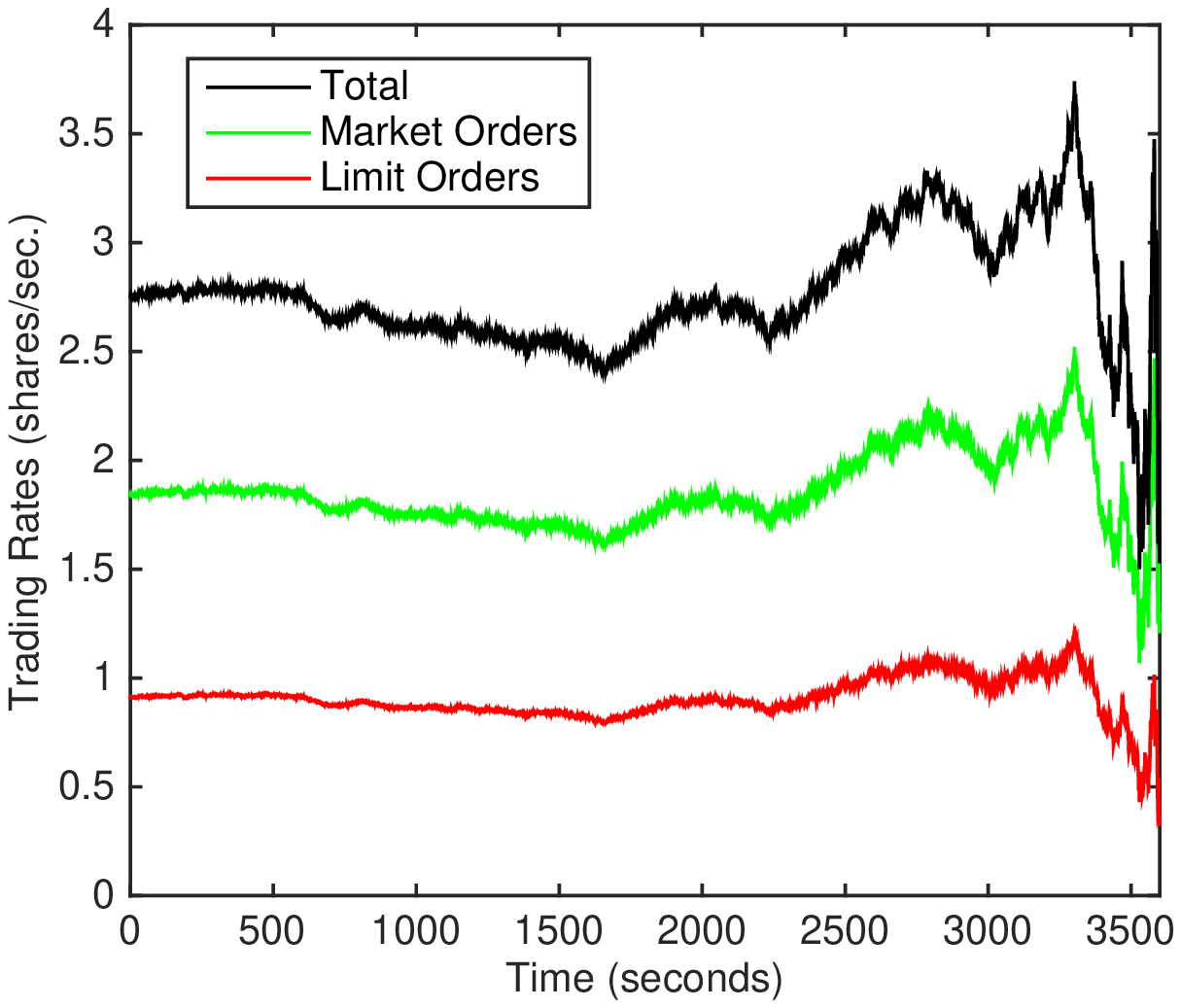}}
	\caption{\small{The optimal trading rates under one simulation with the following parameters: $x_0=10,000$, $T=3,600$, $\beta=10^{-3}$, $\eta_0=0.05$, $\gamma=2.5\cdot10^{-7}$, $\alpha=0.15$, $\beta_1=5\cdot10^{-4}$, $\beta_2=10^{-4}$, $\rho=-0.2$, $\mu=10^{-6}$, $\sigma=0.005$, $m_0=8.\bar{3}$ and $m_1=3$. On the left,  $\eta_1=0.1$, and $\eta_2=0.05$, while on the right  $\eta_1=0.05$, and $\eta_2=0.1$. Final positions are 68.5384, and 80.2371 shares, respectively.}}\label{fig:changing_Impact}
\end{figure}

Next, we look at the effect of   market impact  in Figure \ref{fig:changing_Impact}. We display the simulated trading rates based on  different values of market impact coefficients $(\eta_1, \eta_2)$. On the left panel, market orders have higher temporary market impact ($\eta_1=0.1$, and $\eta_2=0.05$), and on the right, limit orders have higher temporary market impact ($\eta_1=0.05$, and $\eta_2=0.1$). 
As seen in both scenarios, the trading rate is higher for the order type with the lower market impact cost. In practice, market orders are expected to have a higher market impact so $\eta_1>\eta_2$ is the more realistic setting. 

The sample path for stock holdings over time is shown in Figure \ref{fig:variedBeta}. Our trading strategy appears to follow a time-weighted average price (TWAP) strategy. This is evident in the linear and decreasing path of holdings over time. A TWAP strategy seeks to trade constantly through time so that the average realized price is the time-weighted price over the execution period. A linear path of holdings indicates that the total trading rate is approximately constant. In contrast to TWAP, the optimal strategy appears to have a non-zero terminal target for $x_T$. A higher $\beta$ leads the trader to sell more rapidly, and the position ends much closer to zero. Furthermore, the position is generally decreasing, but we remark that even if $v_t>0$, and $L_t>0$, $x_t$ may increase temporarily due to the Brownian motion movement.

\begin{figure}[H]
	\begin{centering}
	\includegraphics[trim=15   0  32  8,clip,width=4in]{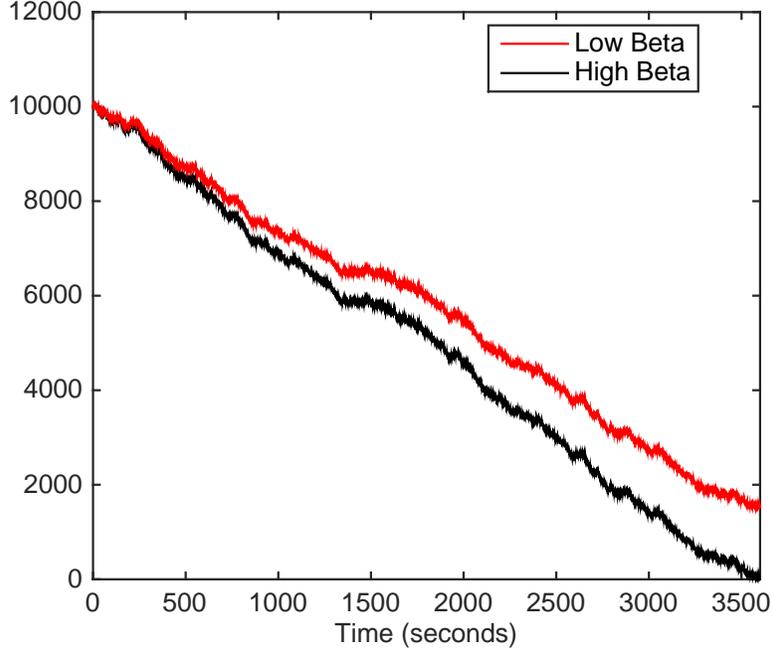}
		\caption{\small{Sample path of stock holdings with different non-liquidation penalty. The parameters are: $x_0=10,000$, $T=3,600$, $\eta_0=0.05$, $\gamma=2.5\cdot10^{-7}$, $\alpha=0.15$, $\beta_1=5\cdot10^{-4}$, $\beta_2=10^{-4}$, $\rho=-0.2$, $\mu=10^{-6}$, $\sigma=0.005$, $\eta_1=0.1$, $\eta_2=0.08$, $m_0=8.\bar{3}$ and $m_1=3$. The low value of $\beta$ is $10^{-4}$ and the high value is $0.1$. Final positions are 1546.0197 and 5.3967, respectively.}}
		\label{fig:variedBeta}
	\end{centering}
\end{figure}

 \section{Constant Uncertainty of Limit Orders}\label{sec:constant_unc}
In this section, we discuss a number of properties of our model under   constant uncertainty of limit orders fills. Recall that constant uncertainty means the total position, $x_t$ is subject to uncertainty and it is not simply due to trading in limit orders. Since the risk in stock holdings cannot be avoided, the tradeoff in choosing market and limit orders is between explicit market impact costs and various implicit costs. The latter costs comprise of the trading penalties set by the trade director and speed limiter discussed in Section \ref{sec:model_order_fills}. Another way to view this section is that $v_t$ could be a trading rate in one exchange, while $L_t$ is a trading rate in another exchange. The two venues have different costs of trading  in terms of market impact, and the total position is subject to risk, measured by constant uncertainty. 

\subsection{Trade Direction-Speed   Trade-off}\label{sec:lagrange}
First,  we must check   the second-order  condition in \eqref{optimalcond}. Since we have $m_1=0$, so the inequality  simplifies to $C>0$, or equivalently
\begin{equation}
\begin{aligned}\label{conditionconstantcase}
\left(\eta_1+\eta_2-\alpha\right)^2-4(\eta_1\eta_2+\beta_1\beta_2+\eta_1\beta_2+\eta_2\beta_1)<0.
\end{aligned}
\end{equation}
Since $\eta_1$, and $\eta_2$ are exogenous parameters to be inferred from market data,    condition \eqref{conditionconstantcase} becomes a restriction on the Lagrange multipliers, $\alpha$, $\beta_1$ and $\beta_2$. Suppose that we fix $\beta_1$ and $\beta_2$. Then the left-hand side  of \eqref{conditionconstantcase} is a convex quadratic function of $\alpha$. Thus, there are 2 roots, and if $\alpha$ is between them, the optimal control problem has a finite solution given uniquely by the first-order conditions. This results in the admissible range for $\alpha$:
\begin{equation}
\begin{aligned}
\alpha\in\left(\eta_1+\eta_2-2\sqrt{\eta_1\eta_2+\beta_1\beta_2+\eta_1\beta_
	2+\eta_2\beta_1},\eta_1+\eta_2+2\sqrt{\eta_1\eta_2+\beta_1\beta_2+\eta_1\beta_2+\eta_2\beta_1}\right).
\end{aligned}
\end{equation}
This leads to our first trade-off. The trader director's Lagrange multiplier, $\alpha$, ensures our order types go in the same direction. On the other hand, the Lagrange multipliers, $\beta_1$ and $\beta_2$, tend to reduce the values of $v_t$ and $L_t$ since high values are less preferable. However, increasing $\beta_i$ widens the allowable interval for $\alpha$. In particular, when $\beta_i=0$, the interval is $\left(\eta_1+\eta_2-2\sqrt{\eta_1\eta_2},\eta_1+\eta_2+2\sqrt{\eta_1\eta_2}\right)$. As the arithmetic mean of a set of numbers is greater than its geometric mean, the left endpoint is positive, which renders  $\alpha=0$ an invalid choice. However, if $\beta_i$ is sufficiently large, then the  interval may include $\alpha=0$. This  discussion reveals the interesting     trade-off between over-reducing trade speed and correctly setting trade direction in  choosing the endogenous parameters, $\alpha$, $\beta_1$ and $\beta_2$. In summary, the  trade director cannot be too strong or too weak unless the speed limiters are sufficiently severe. 

\subsection{Optimal Strategies}\label{sec:cons_strats}
We now consider the associated trading strategies. First, we must solve the ODE system:
\begin{equation}
\begin{aligned}
0&=a^\prime(t)+\frac{2(\alpha+\beta_1+\beta_2)}{(\eta_1+\beta_1)C}a^2(t),\\
0&=b^\prime(t)+\mu+\frac{2(\alpha+\beta_1+\beta_2)}{(\eta_1+\beta_1)C}a(t)(b(t)+\eta_0)+\frac{(\eta_2-\eta_1-\alpha-2\beta_1)\rho\sigma m_0}{(\eta_1+\beta_1)C}a(t),\\
0&=c^\prime(t)+\frac{m_0^2}{2}(2a(t)+\gamma)+\rho\sigma m_0+\beta_1R_1+\beta_2R_2+\frac{(\alpha+\beta_1+\beta_2)}{2(\eta_1+\beta_1)C}(b(t)+\eta_0)^2\\
&+\frac{(\eta_2-\eta_1-\alpha-2\beta_1)\rho\sigma m_0}{2(\eta_1+\beta_1)C}(b(t)+\eta_0)+\frac{(\eta_2-\eta_1-\alpha-2\beta_1)^2\rho^2\sigma^2 m_0^2}{8(\eta_1+\beta_1)C(2(\alpha+\beta_1+\beta_2)-C)},\\
\end{aligned}
\end{equation}
with terminal conditions $a(T)=\frac{\gamma}{2}-\beta$, $b(T)=c(T)=0$.  By separation of variables, we obtain the explicit solution for $a(t)$:
\begin{equation*}
\begin{aligned} 
a(t)=-\frac{(\eta_1+\beta_1)\left(2\beta-{\gamma}\right)C}{2(\alpha+\beta_1+\beta_2)\left(2\beta-{\gamma}\right)(T-t)+2(\eta_1+\beta_1)C}\,.
\end{aligned}
\end{equation*}
The function $a(t)$ is well  defined everywhere except for 
\begin{equation}\label{eq:bad_point}
\begin{aligned} 
t=T+\frac{(\eta_1+\beta_1)C}{(\alpha+\beta_1+\beta_2)\left(2\beta-\gamma\right)}=:T_{crit}.
\end{aligned}
\end{equation}
Suppose that we choose  $\beta$ to be  large in order to penalize non-liquidation, and specifically let us consider $\beta>\frac{\gamma}{2}$.  Then, the second term in \eqref{eq:bad_point} is positive, and $T_{crit}$ is never reached as it is beyond the execution horizon, $T$, of the sell program. With this definition, $a(t)$ simplifies to
\begin{equation*}
\begin{aligned} 
a(t)=-\frac{(\eta_1+\beta_1)C}{2(\alpha+\beta_1+\beta_2)\left(T_{crit}-t\right)}.
\end{aligned}
\end{equation*}
Since, $t\in[0,T]\subsetneq[0,T_{crit}]$, we conclude $a(t)<0$ for $t\in[0,T]$, which implies that $V_{xx}(t,x)<0$ for all $(t,x)\in [0,T]\times\R$. 

To solve for $b(t)$, we divide the second ODE by $a(t)$,  and rearrange to get
\begin{equation}\label{eq:b_constODE}
	\begin{aligned} 
		\frac{b^\prime(t)}{a(t)}+\frac{2(\alpha+\beta_1+\beta_2)}{(\eta_1+\beta_1)C}(b(t)+\eta_0)=-\frac{\mu}{a(t)}-\frac{(\eta_2-\eta_1-\alpha-2\beta_1)\rho\sigma m_0}{(\eta_1+\beta_1)C}.\\
	\end{aligned}
\end{equation}
By using the product rule followed by the ODE for $a(t)$, we have the following identity:
\begin{equation}\label{eq:ODE_id}
	\begin{aligned} 
		\frac{d}{dt}\left[\frac{b(t)+\eta_0}{a(t)}\right]=\frac{b^\prime(t)}{a(t)}-\frac{\left(b(t)+\eta_0\right)a^\prime(t)}{a^2(t)}=\frac{b^\prime(t)}{a(t)}+\frac{2(\alpha+\beta_1+\beta_2)}{(\eta_1+\beta_1)C}(b(t)+\eta_0).
	\end{aligned}
\end{equation}
Define
\begin{equation*}
\begin{aligned} 
b_0(t):=\frac{(\alpha+\beta_1+\beta_2)\mu}{(\eta_1+\beta_1)C}\left[(T_{crit}-t)^2-(T_{crit}-T)^2\right]-\frac{(\eta_2-\eta_1-\alpha-2\beta_1)\rho\sigma m_0(T-t)}{(\eta_1+\beta_1)C}
\end{aligned}
\end{equation*}
to be the integral of the right-hand side function in \eqref{eq:b_constODE} from $t$ to $T$. Then, $b(t)$ is explicitly solved as follows:
\begin{equation*}
	\begin{aligned} 
		\frac{d}{dt}\left[\frac{b(t)+\eta_0}{a(t)}\right]&=-\frac{\mu}{a(t)}-\frac{(\eta_2-\eta_1-\alpha-2\beta_1)\rho\sigma m_0}{(\eta_1+\beta_1)C}\\
		\implies~b(t)&=-\eta_0-a(t)\left[\frac{2\eta_0}{2\beta-\gamma}+b_0(t)\right].
	\end{aligned}
\end{equation*}
Finally, $c(t)$ can be computed  via direct integration. As the integral is rather complicated and we will not need its closed-form expression, we omit it. Direct computation shows that the optimal trading rates are
\begin{align}
v_t^*&=\frac{(\eta_1-\eta_2-2\beta_2-\alpha)(2a(t)x_t+b(t)+\eta_0)}{2(\eta_1+\beta_1)C},\\
L_t^*&=\frac{(\eta_2-\eta_1-\alpha-2\beta_1)(2a(t)x_t+b(t)+\eta_0)}{2(\eta_1+\beta_1)C}.
\end{align}
Notice that the optimal trading strategies are both affine in $x_t$ at any time $t$. We will study the sign of the trading rates in the next section.

%
%
%

\subsection{Buy-Sell Boundary}\label{sec:buySellRegions}
In this section, we describe the properties that guarantee non-negativity of the trading rates for constant uncertainty. Naturally, the optimal trading rates should  be non-negative  because this means trading strategies do not go against the overall sell program. At the very least, if one order rate is non-positive, the other should be non-positive to avoid simultaneous buy and sell orders.

For simplicity, let us assume   $\eta_1=\eta_2\equiv\eta$.\footnote{The results in this section still hold if $\eta_1-\eta_2<2\beta_2+\alpha$ and $\eta_2-\eta_1<2\beta_1+\alpha$. So even if $\eta_1\neq\eta_2$, then one of these conditions certainly holds (because $\beta_1,\beta_2,\alpha>0$ and we either have $\eta_1<\eta_2$ or $\eta_2<\eta_1$.) and the other holds if $|\eta_1-\eta_2|$ is small relative to the Lagrange multipliers.} Then the optimal trading rates are
\begin{align} 
v_t^*&=\frac{\left(-2\beta_2-\alpha\right) (2a(t)x_t+b(t)+\eta_0)}{2(\eta+\beta_1)C},\label{vstar1}\\
L_t^*&=\frac{\left(-2\beta_1-\alpha\right) (2a(t)x_t+b(t)+\eta_0)}{2(\eta+\beta_1)C}.\label{Lstart1}
\end{align}
It follows that $v_t^*, L_t^*>0$ if and only if $2a(t)x_t+b(t)+\eta_0<0$, or equivalently
\begin{equation*}
\begin{aligned} 
2a(t)x_t-a(t)\left[\frac{2\eta_0}{2\beta-\gamma}+b_0(t)\right]<0.
\end{aligned}
\end{equation*}
Since $a(t)<0$ for all $t$, this gives the  lower bound as a  time deterministic function, above which the optimal  strategy derived in \eqref{vstar1} and \eqref{Lstart1} simultaneously places sell orders, and below which the optimal order strategy simultaneously places buy orders. Explicitly, the condition is
\begin{equation*}
\begin{aligned} 
x_t>&\, \frac{(T_{crit}-t)^2(\alpha+\beta_1+\beta_2)\mu}{2(\eta+\beta_1)C}-\frac{(T_{crit}-T)^2(\alpha+\beta_1+\beta_2)\mu}{2(\eta+\beta_1)C}\\
&\quad +\frac{(\alpha+2\beta_1)\rho\sigma m_0(T-t)}{2(\eta+\beta_1)C}+\frac{\eta_0}{2\beta-\gamma}.
\end{aligned}
\end{equation*}
We call this lower boundary on the right-hand side the \textit{buy-sell boundary}, and denote it by $P(t)$. 
\begin{example}\label{ex:inpracticeBSB}
Before discussing the general properties of the buy-sell boundary, we will consider a special case. Often in practice, $\mu$ is assumed to be $0$ as it is unknown. Furthermore, setting $\eta_0=0$, the buy-sell boundary becomes
\begin{equation*}
\begin{aligned} 
P(t)=\frac{(\alpha+2\beta_1)\rho\sigma m_0(T-t)}{2(\eta+\beta_1)C},\quad  t\in[0,T].
\end{aligned}
\end{equation*}
With the adverse selection condition $\rho m_0<0$, $P(t)$ is negative $\forall\, t\in[0,T)$ and increases to the value 0 at $t=T$. Thus, the algorithm continues to sell even if the position becomes short and will only buy if the position becomes substantially negative.
\end{example}

The boundary is a quadratic\footnote{The boundary is quadratic as long as $\mu\neq0$. Otherwise, $P(t)$ is linear increasing (resp. decreasing) if $\rho m_0<0$ (resp. $>0$). For constant uncertainty, $\rho m_0<0$ gives us the desired adverse selection effect as discussed. See Example \ref{ex:inpracticeBSB} for more details.} function of time and is convex (resp. concave) if $\mu>0$ (resp. $\mu<0$). The intuition is that when the position size is sufficiently  low, the trader need not worry about non-liquidation as she has plenty of time to fully liquidate. 

To explore its shape properties further, we compute the first derivative of $P$:
\begin{equation*}
\begin{aligned} 
P^\prime(t)=-\frac{(T_{crit}-t)(\alpha+\beta_1+\beta_2)\mu}{(\eta+\beta_1)C}-\frac{(\alpha+2\beta_1)\rho\sigma m_0}{2(\eta+\beta_1)C}.
\end{aligned}
\end{equation*}
Generally, we expect $P(t)$ to be a non-increasing function of time. To understand when this is the case, let us fix ideas and assume $\mu>0$. Then, $P(t)$ is non-increasing for 
\begin{align} \label{wowineq}
t&\le T_{crit}+\frac{(\alpha+2\beta_1)\rho\sigma m_0}{2(\alpha+\beta_1+\beta_2)\mu}\\
&=T+\frac{(\alpha+2\beta_1)\rho\sigma m_0(2\beta-\gamma)+2(\eta+\beta_1)C\mu}{2\mu(\alpha+\beta_1+\beta_2)(2\beta-\gamma)}.
\end{align}


If there is no adverse selection effect, then $\rho m_0\ge0$, and the second term is positive, so $P$ is non-increasing for all $t\in[0,T]$. However, if there is adverse selection, then $\rho m_0<0$, and it is possible that the second term is negative. However as long as
\begin{equation*}
\begin{aligned} 
|\rho| \le\left|\frac{2(\eta+\beta_1)C\mu}{(\alpha+2\beta_1)\sigma(2\beta-\gamma)m_0}\right|,
\end{aligned}
\end{equation*}
then the second term is non-negative and $P(t)$ is non-increasing for all $t$.
In contrast, to require $P$ to be non-decreasing for all $t$ means reversing the inequality \eqref{wowineq}, that is,
\begin{align}
t&\ge T_{crit}+\frac{(\alpha+2\beta_1)\rho\sigma m_0}{2(\alpha+\beta_1+\beta_1)\mu}\\
&=T+\frac{(\alpha+2\beta_1)\rho\sigma m_0(2\beta-\gamma)+2(\eta+\beta_1)C\mu}{2\mu(\alpha+\beta_1+\beta_2)(2\beta-\gamma)}.\label{wow2}
\end{align}
In order to ensure  that $P(t)$ be  non-decreasing for all $t$, we bound the right-hand side of \eqref{wow2} from above by 0 and rearrange the inequality to obtain
\begin{equation*}
\begin{aligned} 
|\rho| \ge\left|\frac{2(\eta+\beta_1)C\mu+ 2T\mu(\alpha+\beta_1+\beta_2)(2\beta-\gamma)}{(\alpha+2\beta_1)\sigma(2\beta-\gamma)m_0}\right|.
\end{aligned}
\end{equation*}

In Figure \ref{fig:buySellThreshSigmaDec}, we display the buy-sell boundary for different parameter values. As $\beta$ governs aversion to non-liquidation, we analyze the effect of increasing $\beta$ on the boundary. In the left panel, we find that the choice of $\rho=-0.2$ from before leads to a generally increasing $P$. As the more intuitive case was a generally decreasing boundary, we also plot this pair of boundaries for $\rho$ close to zero.\footnote{To have the adverse selection effect, we need $\rho$ to be smaller in absolute value.}

We expect that a larger $\beta$ will induce the trader to place sell orders more often than with a smaller $\beta$. Placing buy   orders takes us further from liquidation so the boundary is expected to shift downward. Indeed this can be seen in both panels of Figure \ref{fig:buySellThreshSigmaDec}. As the whole buy-sell boundary shifts downward, there will of course be a decrease in the terminal value $P(T)$, which can be viewed as the \textit{target} the algorithm has for stock holdings at time $T$. In fact, we have
\begin{equation}
\begin{aligned}
P(T)=\frac{\eta_0}{2\beta-\gamma}.
\end{aligned}
\end{equation}
We can see directly that if the non-liquidation penalty coefficient $\beta$ increases, the target $P(T)$ approaches $0$. This is consistent with  our discussion following Figure \ref{fig:variedBeta}.

\begin{figure}[H]
	\begin{centering}
		\subfigure{\includegraphics[trim = 10 0 35 15, clip, width=2.9in]{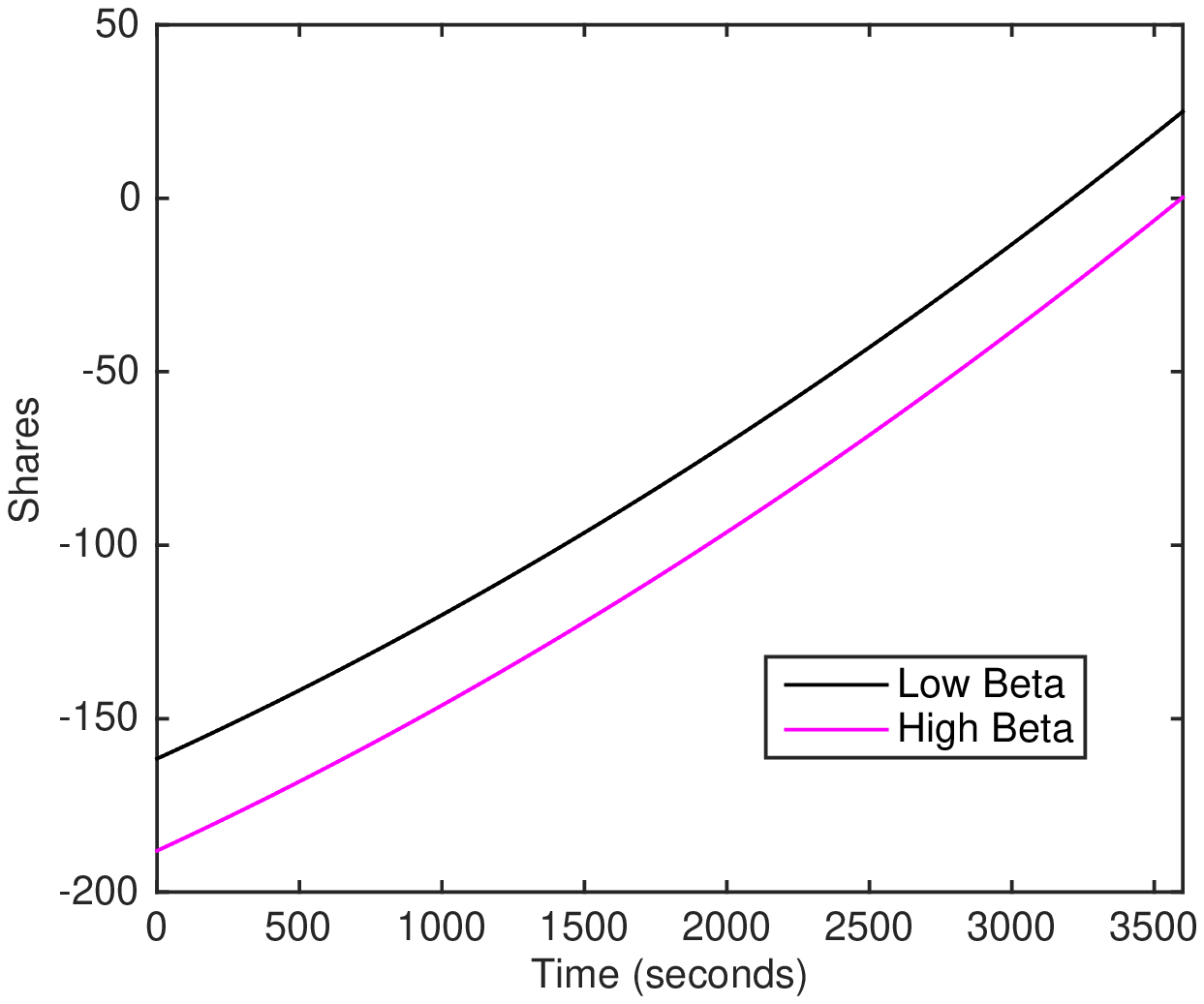}}
		\subfigure{\includegraphics[trim = 10 0 35 15, clip,width=2.9in]{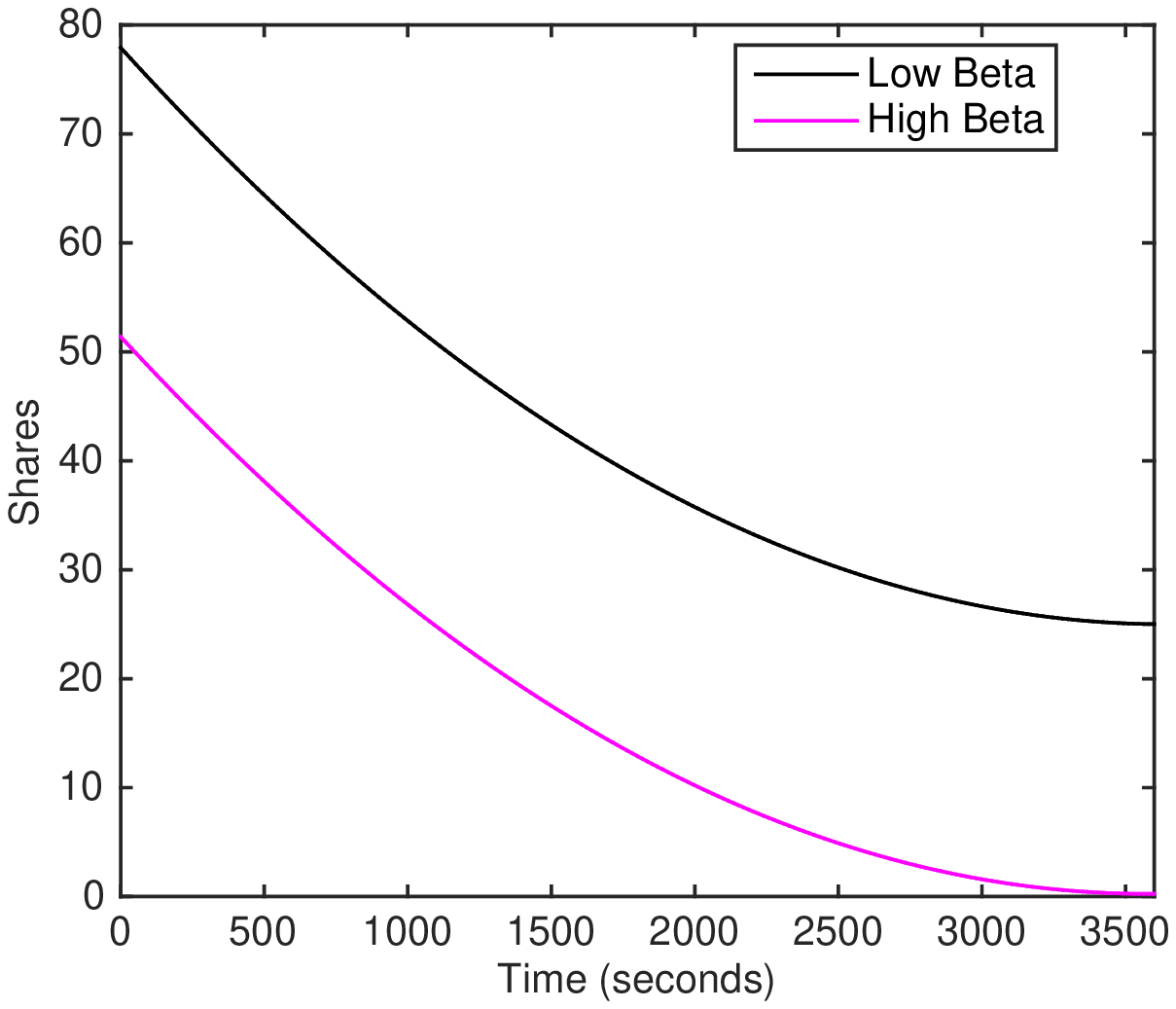}}
		\caption{\small{Buy-sell boundaries over time when $T=3,600$, $\beta_1=5\cdot10^{-4}$ $\beta_2=10^{-4}$, $\eta_0=0.05$, $\gamma=2.5\cdot10^{-7}$, $\eta_1=\eta_2=0.1$, $\mu=10^{-6}$, $\sigma=5\cdot10^{-3}$ and $m_0=16.\bar{6}$. On the left the boundary is generally increasing ($\rho=-0.2$), while on the right, it is generally increasing ($\rho=-0.0005$). We vary $\beta$ in each figure. The low $\beta$ is equal to $10^{-3}$ while the high $\beta$ is $0.1$.}}
		\label{fig:buySellThreshSigmaDec}
	\end{centering}
\end{figure}

\section{Linear Uncertainty of Limit Orders}\label{sec:linear_unc}
In this section, we  consider the case of linear uncertainty of limit orders. This amounts to taking   $m_0=0$ in the affine uncertainty model. We discuss a number of properties of the solution and optimal strategies. 

\subsection{Liquidation Penalty and Trading Horizon Trade-off}\label{sec:penalty_tradeoff}
We begin by writing down the condition for optimality as well as the ODEs for $a$, $b$ and $c$ that characterize the solution to the optimal execution problem. Recall the second-order condition: $V_{xx}(t,x)<\frac{C}{m_1^2}-\gamma$. From the quadratic ansatz, $V_{xx}$ is independent of $x$, so the condition depends only on time. Thus, we find that upon solving the ODE for $a(t)$ numerically or otherwise, we must check that 
\begin{equation*}
\begin{aligned}
\underset{0\le t\le T}{\sup} \, a(t)<\frac{C}{2m_1^2}-\frac{\gamma}{2},
\end{aligned}
\end{equation*}
for the particular set of problem parameters. 

To that end, let us look at the ODEs for $a$, $b$ and $c$. In this case, they solve the system:
\begin{equation}
\begin{aligned}
0&=a^\prime(t)+\frac{2(\alpha+\beta_1+\beta_2)-m_1^2(2a(t)+\gamma)}{(\eta_1+\beta_1)(C-m_1^2(2a(t)+\gamma))}a^2(t),\\
0&=b^\prime(t)+\mu+\frac{2(\alpha+\beta_1+\beta_2)-m_1^2(2a(t)+\gamma)}{(\eta_1+\beta_1)(C-m_1^2(2a(t)+\gamma))}a(t)(b(t)+\eta_0)\\
0&=c^\prime(t)+\beta_1R_1+\beta_2R_2+\frac{2(\alpha+\beta_1+\beta_2)-m_1^2(2a(t)+\gamma)}{4(\eta_1+\beta_1)(C-m_1^2(2a(t)+\gamma))}(b(t)+\eta_0)^2.\\
\end{aligned}
\end{equation}
An implicit equation is available for $a(t)$. However, it is much more enlightening to have an explicit solution. One condition that allows for an explicit solution is
\begin{equation}
\begin{aligned}
2(\alpha+\beta_1+\beta_2)=C.
\end{aligned}
\end{equation}
Although this may seem arbitrary, there are 3 exogenous Lagrange multipliers we can choose freely so it is not too hard to impose this condition. With this restriction, we simplify the equation for $a(t)$ to
\begin{equation}
\begin{aligned}
0=a^\prime(t)+\frac{a^2(t)}{\eta_1+\beta_1},
\end{aligned}
\end{equation}
with the terminal condition $a(T)=\frac{\gamma}{2}-\beta$. This leads to the explicit solution
\begin{equation}\label{atexpli}
a(t)=-\frac{\left(\eta_1+\beta_1\right)\left(2\beta-\gamma\right)}{2\left(\eta_1+\beta_1\right)+(T-t)\left(2\beta-\gamma\right)}.
\end{equation}

Next, we solve for $b(t)$. Divide equation for $b(t)$ by $a(t)$ (again valid since $a(t)<0\,\forall\, t$) and rearrange, we have
\begin{equation}
\begin{aligned}
-\frac{\mu}{a(t)}=\frac{b^\prime(t)}{a(t)}+\frac{b(t)+\eta_0}{\eta_1+\beta_1}.
\end{aligned}
\end{equation}
Plugging in the expression of  $a(t)$ from \eqref{atexpli} and recalling  \eqref{eq:ODE_id}, we arrive at 
\begin{equation}
\begin{aligned}
&\frac{d}{dt}\left[\frac{b(t)+\eta_0}{a(t)}\right]=\frac{2\mu}{2\beta-\gamma}+\frac{\mu(T-t)}{\eta_1+\beta_1}\\
\implies&b(t)=-\eta_0-a(t)\left[\frac{2\eta_0}{2\beta-\gamma}+\frac{2\mu(T-t)}{2\beta-\gamma}+\frac{\mu(T-t)^2}{2(\eta_1+\beta_1)}\right].
\end{aligned}
\end{equation}
 With this,  $c(t)$ can be computed  by direct integration of the associated ODE. The solution is not useful in our analysis however.


Using \eqref{atexpli}, we can express condition \eqref{optimalcond}  that ensures the finiteness of the value function as
\begin{equation}\label{eq:opt_cond_reduce}
\begin{aligned}
-\frac{\left(\eta_1+\beta_1\right)\left(2\beta-\gamma\right)}{2\left(\eta_1+\beta_1\right)+T\left(2\beta-\gamma\right)}<\frac{C-\gamma m_1^2}{2m_1^2}.
\end{aligned}
\end{equation}
Again, we assume $2\beta>\gamma$. If $C\ge\gamma m_1^2$, the left-hand side of \eqref{eq:opt_cond_reduce} is negative, while the right-hand side is non-negative so the condition holds. If on the other hand, $C<\gamma m_1^2$, we must have
\begin{equation}
T<\frac{4m_1^2(\eta_1+\beta_1)(\beta-\gamma)+2(\eta_1+\beta_1)C}{(\gamma m_1^2-C)(2\beta-\gamma)}:=T_{max}.\label{condition11}
\end{equation}
In other words,  we have translated condition \eqref{optimalcond} into   an upper bound on the horizon $T$.  This means that there is a finite maximum trading horizon in order  for the value function to be finite.

Given a fixed $T>0$,  we can turn   condition \eqref{condition11} to a lower bound on $\beta$, that is,   \begin{equation}\label{condccc}
 \beta>\gamma-\frac{C}{2m_1^2}>\frac{\gamma}{2}.
  \end{equation}
In fact, condition \eqref{condccc} is    more stringent than the original condition:  $\beta>\gamma/2$.

Also, the maximum horizon $T_{max}$ is increasing in $\beta$, namely, 
\begin{equation*}
\frac{\partial T_{max}}{\partial\beta}=\frac{\partial}{\partial\beta}\left[\frac{4m_1^2(\eta_1+\beta_1)(\beta-\gamma)+2(\eta_1+\beta_1)C}{(\gamma m_1^2-C)(2\beta-\gamma)}\right]=\frac{4(\eta_1+\beta_1)}{(2\beta-\gamma)^2}>0.
\end{equation*}
This reveals that a higher non-liquidation penalty coefficient $\beta$ permits a longer admissible trading horizon $T_{\max}$ because  the trader is sufficiently motivated to achieve full liquidation, rather than trading for profits.  Nevertheless, there is a finite limit for the maximum horizon. Indeed, as $\beta\rightarrow\infty$, $T_{max}\rightarrow {2m_1^2(\eta_1+\beta_1)}/{(\gamma m_1^2-C)}$.  



Let us consider  the trade-off  differently by imposing a condition on $\beta$. Starting from equation \eqref{eq:opt_cond_reduce}, we obtain
\begin{equation*}
\begin{aligned}
-\left[\eta_1+\beta_1+\frac{C-\gamma m_1^2}{2m_1^2}T\right]\left(\beta-\frac{\gamma}{2}\right)<\frac{C-\gamma m_1^2}{2m_1^2}(\eta_1+\beta_1).
\end{aligned}
\end{equation*}
The coefficient of $\beta-\frac{\gamma}{2}$ must be positive. If it were non-positive then the fact that both $\eta_1>0$ and $\beta_1>0$ implies $C<\gamma m_1^2$ (it must be strict for otherwise, the coefficient would be positive) and so the right-hand side of the inequality is negative. However, the left-hand side would then be non-negative after accounting for the negative sign, so the condition cannot hold if this coefficient is non-positive. It follows that we can rewrite condition \eqref{eq:opt_cond_reduce} as
\begin{equation*}
\begin{aligned}
\beta>\frac{\gamma}{2}-\frac{(C-\gamma m_1^2)(\eta_1+\beta_1)}{2m_1^2(\eta_1+\beta_1)+{(C-\gamma m_1^2)}T}.
\end{aligned}
\end{equation*}
Our discussion above demonstrates that the denominator of the second term is positive does not necessarily require $C\ge\gamma m_1^2$ (the difference can be negative, just not \emph{too} negative), so this is not a trivial condition.\footnote{In other words, if $C\ge\gamma m_1^2$, then (accounting for the negative sign), the lower bound is strictly less than $\frac{\gamma}{2}$. We already require that $\beta>\frac{\gamma}{2}$, so in this case the extra condition is trivial.} Putting this together with the previous restriction on $\beta$, we have
\begin{equation*}
\begin{aligned}
\beta>\frac{\gamma}{2}+\left[\frac{(\gamma m_1^2-C)(\eta_1+\beta_1)}{2m_1^2(\eta_1+\beta_1)+{(C-\gamma m_1^2)}T}\right]^+.
\end{aligned}
\end{equation*}
From this  condition, we see   that the trader must be imposed with a sufficiently high  non-liquidation  penalty. If not, she will spend time profiting from other opportunities over the trading horizon and by time $t=T$, she need not liquidate the asset fully. In particular, she will follow a strategy that exploits profits in the model far in excess of her costs for non-liquidation. Mathematically, this condition on $\beta$ guarantees the finiteness of the value function for the optimal liquidation problem.

\subsection{Infinite Uncertainty Limit}\label{sec:infUncertain}
 
In the case of linear uncertainty, it is possible to set $L_t=0$ and ignore the   limit order fill uncertainty. Conversely, a large linear uncertainty should promote little to no limit orders.   Intuitively, if limit orders have infinite uncertainty to fill, then we expect that they will not be utilized.  To demonstrate this, we  begin by taking the limit as $m_1\rightarrow\infty$ in equation \eqref{eq:ODE_system} with $m_0=0$.\footnote{The ODE is different for $c$ if $m_0\neq0$, but everything else that follows in this section still holds for $a$ and $b$.} For simplicity, we also take $\mu=0$. The limiting ODE system becomes
\begin{equation}
\begin{aligned}
0&=a^\prime(t)+\frac{a^2(t)}{\eta_1+\beta_1},\\
0&=b^\prime(t)+\frac{a(t)(b(t)+\eta_0)}{\eta_1+\beta_1},\\
0&=c^\prime(t)+\beta_1R_1+\beta_2R_2+\frac{(b(t)+\eta_0)^2}{4(\eta_1+\beta_1)},
\end{aligned}
\end{equation}
with the terminal conditions $a(T)=\frac{\gamma}{2}-\beta$, $b(T)=0$, $c(T)=0$.\footnote{As the solution to an ODE is effectively an integral, we are interchanging limit and integral and must justify the switch. The implicit solution to the ODE for $a(t)$ is $-K_1\ln|z_1-a(T)|+K_1\ln|z_1-a(t)|+K_2\ln|a(T)|-K_2\ln|a(t)|-\frac{K_3}{a(T)}+\frac{K_3}{a(t)}=-z_3(T-t)$, where $K_1=K_2=\frac{z_2-z_1}{z_1^2}$, $K_3=\frac{z_2}{z_1}$, $z_1=\frac{2(\alpha+\beta_1+\beta_2)-\gamma m_1^2}{2m_1^2}$, $z_2=\frac{C-\gamma m_1^2}{2m_1^2}$, and $z_3=\frac{1}{\eta_1+\beta_1}$. The limiting solution as $m_1\rightarrow\infty$ is exactly the solution to the above ODE for $a$. The closed form solution for $b$ is $b(t)=-\eta_0+a(t)\left[\int_t^T\frac{\mu}{a(s)}ds-\frac{2\eta_0}{2\beta-\gamma}\right]$. If $\mu=0$, then we can take the limit and get the solution for $b(t)$ from the last section when $\mu=0$ and $m_1\rightarrow\infty$.} We have previously solved the ODE for $a(t)$ in section \ref{sec:penalty_tradeoff}. Therefore, the optimality condition regarding the supremum of $a(t)$ is the same and since $m_1\rightarrow\infty$, we can never have $C\ge\gamma m_1^2$. Thus, there is a condition on $T$ that it must be less than $T_{max}$, when $m_1\rightarrow\infty$. This limit indicates
\begin{equation}
\begin{aligned}
T<\frac{4(\eta_1+\beta_1)(\beta-\gamma)}{\gamma(2\beta-\gamma)}.
\end{aligned}
\end{equation}
We can put this in terms of a lower bound on the non-liquidation penalty:
\begin{equation*}
\begin{aligned}
\beta>\frac{\gamma}{2}+\left[\frac{\gamma(\eta_1+\beta_1)}{2(\eta_1+\beta_1)-\gamma T}\right]^+.
\end{aligned}
\end{equation*}
Assuming the optimality condition holds, then the solution we have generated from the first order conditions is the unique optimizer for the stochastic control problem. Note that implicitly, $\beta$ must be greater than $\gamma$ now for the right-hand side bound to even be positive.
 
Next we consider the trading rates. Looking at equation \eqref{eq:optimzers}, we let $m_1\rightarrow\infty$. When we do that, $L_t^*\rightarrow0$, while
\begin{equation}
\begin{aligned}
v_t^*\rightarrow-\frac{1}{2}\left(\frac{V_x+\eta_0}{\eta_1+\beta_1}\right)=\left(x_t-\frac{\eta_0}{\eta_1+\beta_1}\right)\frac{2\beta-\gamma}{2\left(\eta_1+\beta_1\right)+(T-t)\left(2\beta-\gamma\right)}. 
\end{aligned}
\end{equation}
So in the infinite uncertainty case, the strategy is to place \emph{only} market orders.\footnote{It is worth noting that the trading rate is non-negative as long as $x_t>\frac{\eta_0}{\eta_1+\beta_1}$.} Interestingly, the infinite uncertainty limit is exactly that of \cite{TaiHoACF} in the constant uncertainty case (recall before that this did not depend on $m_0$.) When $\eta_0=0$ and $\beta_1=0$, we get the optimal trading rate derived in the appendix of that paper.

\section{Schedule Following}\label{sec:schedule_following}
In this section, we incorporate a parent order schedule into the order placement problem. The trader  has both market and limit orders at her disposal and is now given  a time-deterministic schedule function, $Q(t)$ defined over the trading horizon $[0,T]$. We assume that $Q(t)$ is a non-negative, non-increasing, bounded continuous function of time, with an initial value $Q(0)=x_0$.   The trader seeks to track $Q(t)$ as closely as possible. Specifically, we would like to keep the stochastic number of shares that we hold at time $t$, $x_t$ close to $Q(t)$ for all times $t\in[0,T]$ and not only at the terminal time $T$. To avoid having conflicting goals, we set $Q(T)=0$ to have a schedule for full liquidation. 


In order to keep $x_t$ close to $Q(t)$, we consider a penalty of the form  $\int_0^T\lambda\left(u,x_u-Q(u)\right)du$, where $\lambda(t,y)$  is a function with global maximum of 0 at $y=0$, $\forall\,t\in[0,T]$. Since $\lambda(t,y)$ has a global max of 0 at 0, the optimizer should choose order types in such a way as to keep $x_t-Q(t)$ close to or equal to 0. Furthermore, the penalty term accumulates deviations at all times $t\in[0,T]$ and may even place more emphasis on certain times due to its time argument. For example, $\lambda(t,y)$ could be increasing in $t$ for all $y$ then deviations early on are allowed, but as we approach the terminal time, the scheduler shall be forced to push $x_t$ closer to $Q(t)$. Henceforth, we let  $\lambda(t,y)=-w(t)y^2$ so that deviations from above/below are penalized equally, and in just the same manner as the terminal penalty. We assume that $w(t)$ is a non-negative and continous function of time. Moreover, we assume that $w(t)$ is bounded over the interval $[0,T].$


Now we will maximize the sum of expected compensated PNL as defined previously together with expected accumulated deviations. The value function is \begin{equation}
\begin{aligned}
&V(t,x):=\underset{(v_t,L_t)_{0\le t\le T}}{\sup}\,\E\left[\bar{f}(x)+\frac{\gamma}{2}x_T^2+\int_t^T \left[g(x_u,L_u,v_u) +\lambda(u,x_u-Q(u))\right]du\Big|x_t=x\right]-\frac{\gamma}{2}x^2.
\end{aligned}
\end{equation}
We then conclude that $V$ satisfies the following nonlinear HJB PDE problem:
\begin{equation}
\begin{aligned}
V_t+\underset{v,L}{\sup}\left[-(v+L)V_x+\frac{1}{2}m^2(L)V_{xx}+g(x,v,L)\right]+\lambda(t,x-Q(t)) &= 0, \quad (t,x) \in [0,T) \times \R,\\
V(T,x)&=\bar{f}(x)+\frac{\gamma}{2}x^2, \quad x \in \R.
\end{aligned}
\end{equation}
Like in previous sections, one can perform the same optimization to derive the optimal trading rates $v^*$ and $L^*$, and the value function will be of the same quadratic form: $V(t,x)=a(t)x^2+b(t)x+c(t)$. However, the inhomogeneous term $\lambda(t,x-Q(t))$ will affect the solutions of resulting ODEs. Specifically, we have the ODE system:
\begin{equation}
\begin{aligned}
0&=a^\prime(t)+\frac{2(\alpha+\beta_1+\beta_2)-m_1^2(2a(t)+\gamma)}{(\eta_1+\beta_1)(C-m_1^2(2a(t)+\gamma))}a^2(t)-w(t),\\
0&=b^\prime(t)+\mu+\frac{2(\alpha+\beta_1+\beta_2)-m_1^2(2a(t)+\gamma)}{(\eta_1+\beta_1)(C-m_1^2(2a(t)+\gamma))}a(t)(b(t)+\eta_0)\\
&+\frac{(\eta_2-\eta_1-\alpha-2\beta_1)(m_0m_1(2a(t)+\gamma)+\rho\sigma m_0)}{(\eta_1+\beta_1)(C-m_1^2(2a(t)+\gamma))}a(t)+2w(t)Q(t),\\
0&=c^\prime(t)+\frac{m_0^2}{2}(2a(t)+\gamma)+\rho\sigma m_0+\beta_1R_1+\beta_2R_2+\frac{2(\alpha+\beta_1+\beta_2)-m_1^2(2a(t)+\gamma)}{4(\eta_1+\beta_1)(C-m_1^2(2a(t)+\gamma))}(b(t)+\eta_0)^2\\
&+\frac{(\eta_2-\eta_1-\alpha-2\beta_1)(m_0m_1(2a(t)+\gamma)+\rho\sigma m_0)}{2(\eta_1+\beta_1)(C-m_1^2(2a(t)+\gamma))}(b(t)+\eta_0)\\
&+\frac{(\eta_2-\eta_1-\alpha-2\beta_1)^2(m_0m_1(2a(t)+\gamma)+\rho\sigma m_0)^2}{8(\eta_1+\beta_1)(C-m_1^2(2a(t)+\gamma))(2(\alpha+\beta_1+\beta_2)-C)}-w(t)Q(t)^2\\
\end{aligned}
\end{equation}
with terminal conditions $a(T)=\frac{\gamma}{2}-\beta$, $b(T)=c(T)=0$.

\begin{figure}[H]
	\begin{centering}
		\includegraphics[trim = 15 0 25 15, clip,  width=4in]{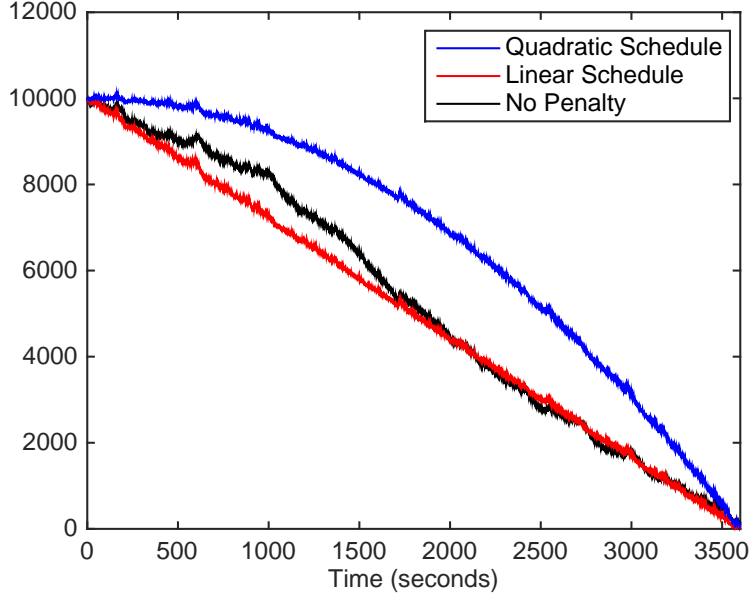}
		\caption{\small{The improvement to schedule following with a small, time-uniform penalty $w(t)=10^{-4}\,\forall\,t$. Other parameters are $x_0=10,000$, $T=3,600$, $\beta_1=5\cdot10^{-4}$, $\beta_2=10^{-4}$, $\eta_0=0.05$, $\gamma=2.5\cdot10^{-7}$, $\eta_1=0.1$, $\eta_2=0.08$, $\mu=10^{-6}$, $\sigma=0.005$, $\rho=-0.2$, $m_0=8.\bar{3}$, and $m_1=3$. Superimposed are an unpenalized time series (blue), a time series targeting TWAP (red) and a time series targeting $Q(t)=x_0\left[1-\left(\frac{t}{T}\right)^2\right]$.}}
		\label{fig:improvementFigSec5_combined}
	\end{centering}
\end{figure}


In practice,  traders may be given guidance to follow the time-deterministic schedule, such as the one     in  \cite{AC2000}. Our   model   allows a trader to quantitatively evaluate the cost of deviating from schedule. Here, we illustrate the  allocation of  market and limit orders over time accounting for the penalty of schedule deviation.    Figure \ref{fig:improvementFigSec5_combined} displays  the trader's stock holdings over time in three settings: (i) a quadratic schedule defined by $Q(t)=x_0\left[1-\left(\frac{t}{T}\right)^2\right]$, and (ii) a linear schedule $Q(t) = x_0\left(1-\frac{t}{T}\right)$, and (iii) no schedule (wherein we set $w(t)\equiv0$). As we can see, the trader's position persistently tracks the schedule over the trading horizon, even with a small  constant penalization coefficient $w(t)=10^{-4}$. Compared to the linear schedule, the quadratic schedule is useful for an institution that seeks to start out trading slowly and eventually speed up at the end of the sell program. We see that our model is capable of handling complicated non-linear schedules for stock holdings over time. 

We examine   the effects on trading rates in Figure \ref{fig:scheduledTradingRates}. On the left panel, there is no penalty for deviating from the linear schedule, while on the right we give a small penalty for deviating. We find that the penalized allocator trades much more quickly and reacts more quickly to price movements. Indeed, the non-penalized trading rates are very stable until the very end of the sell program and the penalized trading rates change rapidly over time. In contrast, the total trading rate in the non-penalized case (right panel) is fluctuating in a relatively small range. In the penalized case, more market orders are used over time, but the opposite is true when the penalty is removed. 

\begin{figure}[H]
	\begin{centering}
		\subfigure{\includegraphics[width=3in]{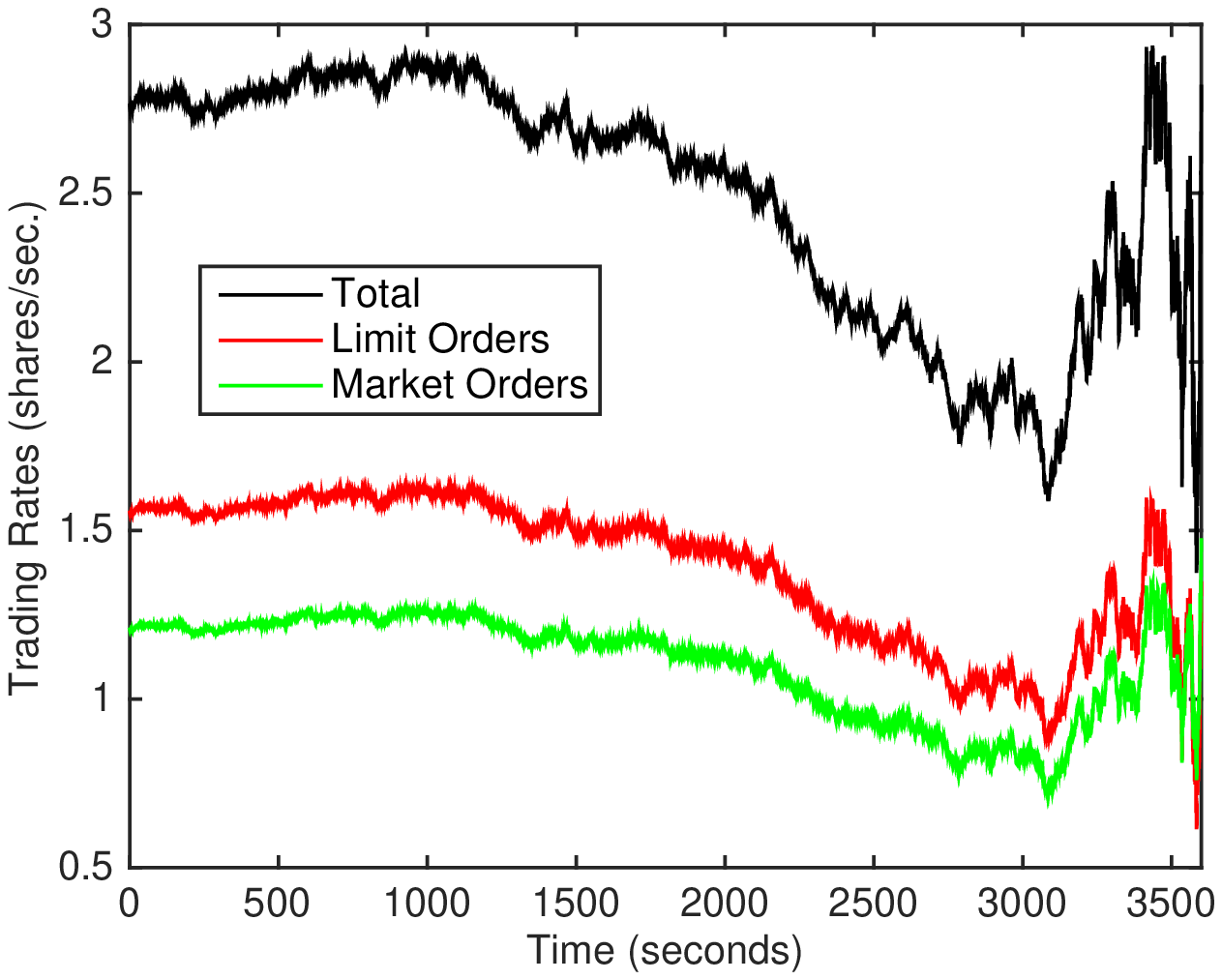}}
		\subfigure{\includegraphics[width=3in]{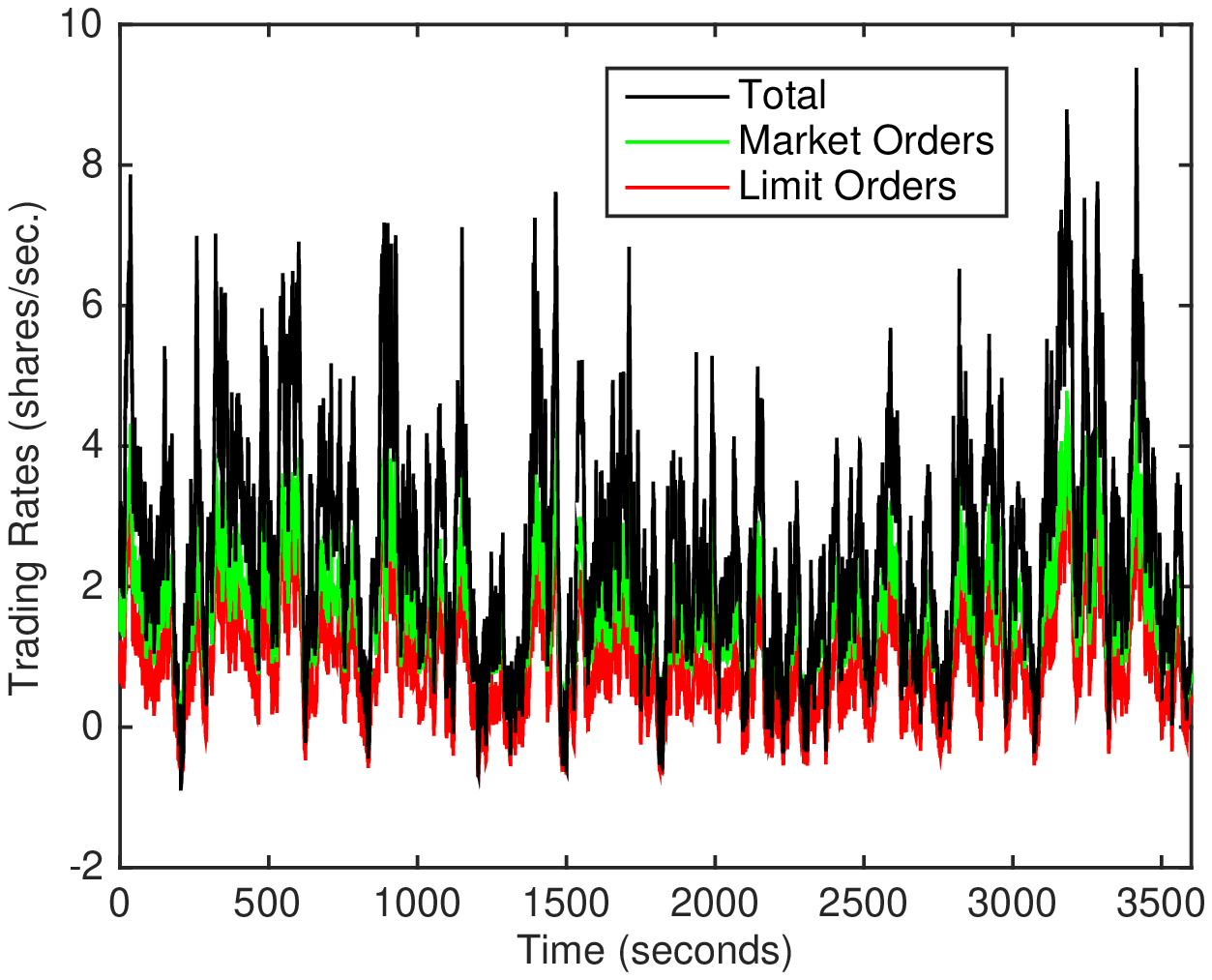}}
		\caption{\small{Comparison of trading rates with and without penalty. On the left, there is no penalty and on the left we penalize deviations from a linear schedule with constant weight $w(t)=10^{-4}$. Other parameters are $x_0=10,000$, $T=3,600$, $\beta_1=5\cdot10^{-4}$, $\beta_2=10^{-4}$, $\eta_0=0.05$, $\gamma=2.5\cdot10^{-7}$, $\eta_1=0.1$, $\eta_2=0.08$, $\mu=10^{-6}$, $\sigma=0.005$, $\rho=-0.2$, $m_0=8.\bar{3}$, and $m_1=3$. Final positions are 190.4408 and 94.4218, respectively.}}
		\label{fig:scheduledTradingRates}
	\end{centering}
\end{figure}

%

\section{Concluding Remarks}\label{sec:conclude} By incorporating a number of new features, our order placement model allows for better control of the trading problem. For example, the trade limiter and director regulate  the speed and direction of trades. Moreover, the multiple penalties lead to a number of trade-offs. Among them,  we find that the trade director cannot be too strong or too weak unless the speed limiters are sufficiently severe. 

The literature on optimal execution continues to grow in volume and diversity.  One direction  for future research as pertained to our  order placement problem  is  to expand the trading platform to include multiple exchanges to allocate orders as in \cite{arseniythesis}. These exchanges have varying costs of filling orders and risks associated with them. This gives rise to a venue allocation problem. Another related problem is trading in lit and dark pools. Whereas in a lit pool (or exchange) the order book imbalance can be observed, a dark pool blinds the trader from such information (see e.g. \cite{kratzSchoneborn}). This suggests a robust optimization or model uncertainty approach to optimal execution not previously considered. Furthermore,  the problem of  optimal liquidation is applicable not only to stocks, but also to options and other derivatives (see e.g. \cite{LeungLiu2012,LeungShirai2015}). 

\singlespacing

\begin{small}
    \bibliographystyle{apa}
    \bibliography{mybib}

\begin{thebibliography}{}

\bibitem[\protect\astroncite{Almgren and Chriss}{2000}]{AC2000}
Almgren, R. and Chriss, N. (2000).
\newblock Optimal execution of portfolio transactions.
\newblock {\em Journal of Risk}, 3:5--39.

\bibitem[\protect\astroncite{Avellaneda and
  Stoikov}{2008}]{AvellandaStoikov2008}
Avellaneda, M. and Stoikov, S. (2008).
\newblock High-frequency trading in a limit order book.
\newblock {\em Quantitative Finance}, 8(3):217--224.

\bibitem[\protect\astroncite{Cartea and Jaimungal}{2015}]{SeblimitmarketQF2015}
Cartea, A. and Jaimungal, S. (2015).
\newblock Optimal execution with limit and market orders.
\newblock {\em Quantitative Finance}, 13(8):1279--1291.

\bibitem[\protect\astroncite{Cartea and Jaimungal}{2016}]{sebVWAP}
Cartea, A. and Jaimungal, S. (2016).
\newblock A closed-form execution strategy to target volume weighted average
  price.
\newblock {\em SIAM Journal of Financial Mathematics}, 7:760–785.

\bibitem[\protect\astroncite{Cartea et~al.}{2015}]{sebBook}
Cartea, A., Jaimungal, S., and Penalva, J. (2015).
\newblock {\em Algorithmic and High-Frequency Trading}.
\newblock Cambridge University Press, 1st edition.

\bibitem[\protect\astroncite{Cheng et~al.}{2017}]{TaiHoACF}
Cheng, X., {Di Giacinto}, M., and Wang, T. (2017).
\newblock Optimal execution with uncertain order fills in {A}lmgren-{C}hriss
  framework.
\newblock {\em Quantitative Finance}, 17(1):55--69.

\bibitem[\protect\astroncite{Guilbaud and Pham}{2013}]{GuilbaudPham2013}
Guilbaud, F. and Pham, H. (2013).
\newblock Optimal high-frequency trading with limit and market orders.
\newblock {\em Quantitative Finance}, 15(1):79--94.

\bibitem[\protect\astroncite{Kratz and Schoeneborn}{2015}]{kratzSchoneborn}
Kratz, P. and Schoeneborn, G. (2015).
\newblock Portfolio liquidation in dark pools in continuous time.
\newblock {\em Mathematical Finance}, 25(3):496--544.

\bibitem[\protect\astroncite{Kukanov}{2013}]{arseniythesis}
Kukanov, A. (2013).
\newblock {\em Stochastic Models of Limit Order Markets}.
\newblock PhD thesis, Columbia University.

\bibitem[\protect\astroncite{Lee}{2008}]{lee2008}
Lee, K. (2008).
\newblock Risk minimization under budget constraints.
\newblock {\em The Journal of Risk Finance}, 9(1):71--80.

\bibitem[\protect\astroncite{Lehalle and Laruelle}{2013}]{lehalleBook}
Lehalle, C. and Laruelle, S. (2013).
\newblock {\em Market Microstructure in Practice}.
\newblock World Scientific Publishing Company, 1st edition.

\bibitem[\protect\astroncite{Leung and Liu}{2012}]{LeungLiu2012}
Leung, T. and Liu, P. (2012).
\newblock Risk premia and optimal liquidation of credit derivatives.
\newblock {\em International Journal of Theoretical \& Applied Finance},
  15(8):1250059.

\bibitem[\protect\astroncite{Leung and Shirai}{2015}]{LeungShirai2015}
Leung, T. and Shirai, Y. (2015).
\newblock Optimal derivative liquidation timing under path-dependent risk
  penalties.
\newblock {\em Journal of Financial Engineering}, 2(1):1550004.

\end{thebibliography}
\end{small}

\end{document}